%% file: main.tex
\documentclass[]{article}
\usepackage{amssymb,calc,enumerate,booktabs}
\usepackage{multirow}
\usepackage{etoolbox,rotating,graphicx}
\usepackage{savesym}
\savesymbol{AND}
\usepackage{algorithm}
\usepackage{algpseudocode}
\usepackage{caption}
\usepackage{url}
\usepackage{savesym}
\usepackage{wrapfig}
\usepackage{amsmath}
\usepackage[subrefformat=parens,labelformat=parens]{subcaption}
\usepackage{comment}
\usepackage{amsthm}
\usepackage{arydshln}
\usepackage{color}
\usepackage{cite}
\usepackage{geometry}
\usepackage{pdflscape}
\usepackage{booktabs}
\usepackage{multicol}
\allowdisplaybreaks

\algnewcommand\Input{\item[{\textbf{Input:}}]}
\algnewcommand\Output{\item[{\textbf{Output:}}]}

\DeclareGraphicsExtensions{.pdf,.jpeg,.png}
\newtheorem{theorem}{Theorem}

\newcommand{\Aa}{\mathcal{A}}
\newcommand{\xx}{\mathcal{X}}
\newcommand{\yy}{\mathcal{Y}}

\newcommand{\ff}{\mathcal{F}}

\newcommand{\mb}[1]{\mathbf{#1}}

\usepackage{xcolor}
\usepackage{soul}

\begin{document}

\title{Bilinear Assignment Problem: Large Neighborhoods and Experimental Analysis of Algorithms}

\author{\sc{Vladyslav Sokol}\thanks{{\tt vsokol@sfu.ca}.
School of Computing Science, Simon Fraser University, 8888 University Drive, Burnaby, British Columbia, V5A 1S6, Canada}
\and
\sc{Ante \'Custi\'c}\thanks{{\tt acustic@sfu.ca}.
Department of Mathematics, Simon Fraser University Surrey, 250-13450 102nd AV, Surrey, British Columbia, V3T 0A3, Canada}
\and
\sc{Abraham P. Punnen}\thanks{{\tt apunnen@sfu.ca}. Department of Mathematics, Simon Fraser University Surrey,  250-13450 102nd AV, Surrey, British Columbia, V3T 0A3, Canada}
\and
\sc{Binay Bhattacharya}\thanks{{\tt binay@sfu.ca}.
School of Computing Science, Simon Fraser University, 8888 University Drive, Burnaby, British Columbia, V5A 1S6, Canada}}

\maketitle

\begin{abstract}
The \emph{bilinear assignment problem (BAP)} is a generalization of the well-known \emph{quadratic assignment problem (QAP)}.
In this paper, we study the problem from the computational analysis point of view. Several classes of neigborhood structures are introduced for the problem along with some theoretical analysis. These neighborhoods are then explored within a local search and a variable neighborhood search frameworks with multistart to generate robust heuristic algorithms. Results of systematic experimental analysis have been presented which divulge the effectiveness of our algorithms. In addition, we present several very fast construction heuristics. Our experimental results disclosed some interesting properties of the BAP model, different from those of comparable models. This is the first thorough experimental analysis of algorithms on BAP. We have also introduced benchmark test instances that can be used for future experiments on exact and heuristic algorithms for the problem.

\medskip

\noindent\emph{Keywords:} bilinear assignment problem, quadratic assignment problem, average solution value, exponential neighborhoods, heuristics, local search, variable neighborhood search, VLSN search.

\end{abstract}

%%%%%%%%%%%%%%%%%%%%%%%%%%%%%%%%%%%%%%%%%%%%%%%%%%%%%%%%%%%%%%%%
%%%%%%%%%%%%%%%%%%%%%%%%%%%%%%%%%%%%%%%%%%%%%%%%%%%%%%%%%%%%%%%%
%%%%%%%%%%%%%%%%%%%%%%%%%%%%%%%%%%%%%%%%%%%%%%%%%%%%%%%%%%%%%%%%
%%%%%%%%%%%%%%%%%%%%%%%%%%%%%%%%%%%%%%%%%%%%%%%%%%%%%%%%%%%%%%%%
\section{Introduction}
\label{sec:intro}

Given a four dimensional array $Q=(q_{ijkl})$ of size $m\times m\times n\times n$, an $m\times m$ matrix $C=(c_{ij})$ and an $n\times n$ matrix $D=(d_{kl})$, the {\it bilinear assignment problem} (BAP) can be stated as:
\begin{align}
	\text{Minimize} \qquad &\sum_{i=1}^m\sum_{j=1}^m\sum_{k=1}^n\sum_{l=1}^n q_{ijkl}x_{ij}y_{kl} + \sum_{i=1}^m\sum_{j=1}^m c_{ij}x_{ij} + \sum_{k=1}^n\sum_{l=1}^n d_{kl}y_{kl} \label{of}\\
	\text{subject to}\quad \  \ & \sum_{j=1}^m x_{ij}=1 \qquad \qquad i=1,2,\ldots,m, \label{x1}\\
	&\sum_{i=1}^m x_{ij}=1 \qquad \qquad j=1,2,\ldots,m, \label{x2}\\
	&\sum_{l=1}^n y_{kl}=1 \qquad \qquad k=1,2,\ldots,n, \label{y1}\\
	&\sum_{k=1}^n y_{kl}=1 \qquad \qquad l=1,2,\ldots,n, \label{y2}\\
	&x_{ij},\ y_{kl}\in \{0,1\} \qquad  i,j=1,\ldots,m,\ \ k,l=1,\ldots,n. \label{int}
\end{align}

If we impose additional restrictions that $m=n$ and $x_{ij}=y_{ij}$ for all $i,j$, BAP becomes equivalent to the well-known \emph{quadratic assignment problem} (QAP) \cite{burkard2012assignment,cela2013quadratic}. As noted in \cite{custicbilinear}, the constraints $x_{ij}=y_{ij}$ can be enforced without explicitly stating them by modifying the entries of $Q,$ $C$ and $D$. For example, replacing $c_{ij}$ by $c_{ij}+L$, $d_{ij}$ by $d_{ij}+L$ and $q_{ijij}$ by $q_{ijij}-2L$, for some large $L$ results in an increase in the objective function value by $\sum_{i,j=1}^n L(x_{ij}-2x_{ij}y_{ij}+y_{ij})=\sum_{i,j=1}^nL(x_{ij}-y_{ij})^2$. Since $L$ is large, in an optimal solution, $x_{ij}=y_{ij}$ is forced and hence the modified BAP becomes QAP. Therefore, BAP is also strongly NP-hard.  Moreover, since the reduction described above preserves the objective values of the solutions that satisfy $x_{ij}=y_{ij}$, BAP inherits the approximability hardness of  QAP \cite{sahni1976p}. That is, for any $\alpha > 1$, BAP does not have a polynomial time $\alpha$-approximation algorithm, unless P=NP. Further, BAP is NP-hard even if $m=n$ and $Q$ is a diagonal matrix \cite{custicbilinear}. A special case of BAP, called the independent quadratic assignment problem, was studied by Burkard et al.~\cite{burkard1998quadratic} and identified polynomially solvable special cases.

Since BAP is a generalization of the QAP, all of the applications of QAP can be solved as BAP. In addition, BAP can be used to model other discrete optimization problems with practical applications. Tsui and Chang \cite{tsui1990microcomputer,tsui1992optimal} used BAP to model a door dock assignment problem. Consider a sorting facility of a large shipping company where $m$ loaded inbound trucks are arriving from different locations, and they need to be assigned to $m$ inbound doors of the facility. The shipments from the inbound trucks need to be transferred to  $n$ outbound trucks, which carries the shipments to  different customer locations. The sorting facility also has $n$ outbound doors for the outbound trucks. Let $w_{ij}$ denote the amount of items from $i$-th inbound truck that need to be transferred to $j$-th outbound truck/customer location, and let $d_{ij}$ denote the distance between the $i$-th inbound door and the $j$-th outbound door. Then the problem of assigning inbound trucks to inbound doors and outbound trucks to outbound doors, so that the total work needed to transfer all items from inbound to outbound trucks, is exactly BAP with costs $q_{ijkl}=w_{ik}d_{jl}$. Torki et al.~\cite{torki1996low} used BAP to develop heuristic algorithms for QAP with a low rank cost matrix. BAP also encompasses well-known disjoint matching problem \cite{custicbilinear,fon1997arrays,frieze1983complexity} and axial 3-dimensional assignment problem \cite{custicbilinear,pierskalla1968letter}.

Despite the applicability and unifying capabilities of the model, BAP is not studied systematically from an experimental analysis point of view. In \cite{tsui1990microcomputer,tsui1992optimal}, the authors proposed  local search and branch and bound algorithms to solve BAP, but detailed computational analysis was not provided. The model was specially structured to focus on a single application, which limited the applicability of these algorithms for the general case. Torki et al.~\cite{torki1996low} presented experimental results on algorithms for low rank BAP in connection with developing heuristics for QAP. To the best of our knowledge, no other experimental studies on the model are available.

\medskip

In this paper, we present various neighborhoods associated with a feasible solution of BAP and analyze their theoretical properties in the context of local search algorithms, particularly on the worst case behavior. Some of these neighborhoods are of exponential size but can be searched for an improving solution in polynomial time. Local search algorithms with such \textit{very large scale neighborhoods (VLSN)} proved to be an effective solution approach for many hard combinatorial optimization problems \cite{ahuja2002survey,ahuja2007very}. We also present extensive experimental results by embedding these neighborhoods within a \textit{variable neighborhood search (VNS)} framework in addition to the standard and multi-start VLSN local search. Some very fast construction heuristics are also provided along with experimental analysis. Although local search and variable neighborhood search are well known algorithmic paradigms that are thoroughly investigated in the context of various combinatorial optimization problems, to achieve effectiveness and obtain superior outcomes variable neighborhood search algorithms needs to exploit special problem structures that efficiently link the various neighborhoods under consideration. In this sense, developing variable neighborhood search algorithms is always intriguing, especially when it comes to new optimization problems having several well designed neighborhood structures with interesting properties.
Our experimental analysis shows that the average behavior of the algorithms are much better and the established negative worst case performance hardly occurs. Such a conclusion can only be made by systematic experimentation, as we have done. On a balance of computational time and solution quality, a multi-start based VLSN local search became our proposed approach. Although, by allowing significantly more time, a strategic variable neighborhood search outperformed this algorithm in terms of solution quality.

\smallskip

The rest of the paper is organized as follows.
In Section \ref{sec:notations} we specify notations and several relevant results that are used in the paper.
In Section \ref{sec:constr} we describe several construction heuristics for BAP that generate reasonable solutions, often quickly.
In Section \ref{sec:ls}, we present various neighborhood structures and analyze their theoretical properties.
We then (Section \ref{sec:expsetup}) describe in details specifics of our experimental setup as well as sets of instances that we have generated for the problem.
%The benchmark instances that we have developed are available at \url{https://vault.sfu.ca/?} for other researchers to further study of the problem.
The benchmark instances that we have developed are available upon request from Abraham Punnen (apunnen@sfu.ca) for other researchers to further study the problem.
The development of these test instances and best-known solutions is yet another contribution of this work.
Sections \ref{sec:expconstr} and \ref{sec:expls} deal with experimental analysis of construction heuristics and local search algorithms. Our computational results disclose some interesting and unexpected outcomes, particularly when comparing standard local search with its multi-start counterpart. 
In Section \ref{sec:expvnsms} we combine better performing construction heuristics and different local search algorithms to develop several variable neighborhood search algorithms and present comparison with our best performing multistart local search algorithm.
Concluding remarks are presented in Section \ref{sec:conclusion}.

%%%%%%%%%%%%%%%%%%%%%%%%%%%%%%%%%%%%%%%%%%%%%%%%%%%%%%%%%%%%%%%%
%%%%%%%%%%%%%%%%%%%%%%%%%%%%%%%%%%%%%%%%%%%%%%%%%%%%%%%%%%%%%%%%
%%%%%%%%%%%%%%%%%%%%%%%%%%%%%%%%%%%%%%%%%%%%%%%%%%%%%%%%%%%%%%%%
%%%%%%%%%%%%%%%%%%%%%%%%%%%%%%%%%%%%%%%%%%%%%%%%%%%%%%%%%%%%%%%%
\section{Notations and basic results}
\label{sec:notations}

Let $\xx$ be the set of all 0-1 $m\times m$  matrices satisfying \eqref{x1} and \eqref{x2} and $\yy$ be the set of all 0-1 $n \times n$ matrices satisfying \eqref{y1} and \eqref{y2}. Also, let $\ff$ be the set of all feasible solutions of BAP. Note that $|\ff|=m!n!$. An instance of the BAP is completely represented by the triplet  $(Q,C,D)$. Let $M=M'=\{1,2,\ldots,m\}$ and $N=N'=\{1,2,\ldots,n\}$. An $\mb{x}\in \xx$ assigns each $i\in M$ a unique $j\in M'$. Likewise, a $\mb{y}\in\yy$ assigns each $k\in N$ a unique $l \in N'$. Without loss of generality we assume that $m\leq n$. For $\mb{x}\in\xx$ and $\mb{y}\in\yy$, $f(\mb{x},\mb{y})$ denotes the objective function value of ($\mb{x},\mb{y}$).

Given an instance $(Q,C,D)$ of a  BAP, let $\Aa(Q,C,D)$ be the average of the objective function values of all feasible solutions.

\begin{theorem}[\'Custi\'c et al.~\cite{custicbilinear}]
\label{thm:Aa}
 ${\displaystyle 	\Aa(Q,C,D)=\frac{1}{mn}\sum_{i=1}^m\sum_{j=1}^m\sum_{k=1}^n\sum_{l=1}^n q_{ijkl}+ \frac{1}{m}\sum_{i=1}^m\sum_{j=1}^mc_{ij}+ \frac{1}{n}\sum_{k=1}^n\sum_{l=1}^nd_{kl}}$.
\end{theorem}

\smallskip

Consider an equivalence relation $\sim$ on $\ff$, where $(\mb{x},\mb{y})\sim(\mb{x}',\mb{y}')$ if and only if there exist $a\in\{0,1,\ldots,m-1\}$ and $b\in\{0,1,\ldots,n-1\}$ such that $x_{ij}=x'_{i(j+a \mod m)}$ for all $i,j$, and $y_{kl}=y'_{k(l+b \mod n)}$ for all $k,l$.
Here and later in the paper we use the notation of $x_{i(j+a \mod m)}$ in a sense that, if $(j+a) \mod m=0$, we then assume it to refer to the variable $x_{im}$. Similar assumptions will be made for the other index of $x_{ij}$ and variables $y_{kl}$ to improve the clarity of presentation.

Let us consider an example of equivalence class for $\sim$. Given $a\in M$, $b\in N$ let $(\mb{x}^a,\mb{y}^b)\in\ff$ be defined as
\[
x_{ij}^a=
\begin{cases}
	1 & \text{if } j=i+a \mod m, \\
	0 & \text{otherwise}
\end{cases}\quad \text{and}\quad
y_{kl}^b=
\begin{cases}
	1 & \text{if } l=k+b \mod n, \\
	0 & \text{otherwise}.
\end{cases}
\]

\begin{theorem}[\'Custi\'c et al.~\cite{custicbilinear}]
\label{thm:minmax}
	For any instance $(Q,C,D)$ of BAP
	\[\min_{a\in M,b\in N}\{f(\mb{x}^a,\mb{y}^b)\}\leq\Aa(Q,C,D)\leq\max_{a\in M,b\in N}\{f(\mb{x}^a,\mb{y}^b)\}.\]
\end{theorem}
It can be shown that any equivalence class defined by $\sim$ can be used to obtain the type of inequalities stated above. Theorem \ref{thm:minmax} provides a way to find a feasible solution to BAP with objective function value no worse than $\Aa(Q,C,D)$ in $O(m^2n^2)$ time.
To achieve this, we search through the set of solutions defined by the equivalence class, with any feasible solution to BAP as a starting point.

\smallskip

A feasible solution $(\mb{x},\mb{y})$ to BAP is said to be no better than the average if $f(\mb{x},\mb{y}) \geq A(Q,C,D)$. In \cite{custicbilinear} we have provided the following lower bound for the number of feasible solutions that are no better than the average.

\begin{theorem}[\'Custi\'c et al.~\cite{custicbilinear}]
\label{thm:dom}
 $|\{(\mb{x},\mb{y})\in\ff\ \colon f(\mb{x},\mb{y})\geq A(Q,C,D)\}|\geq (m-1)!(n-1)!$.
\end{theorem}

An algorithm that is guaranteed to return a solution with the objective function value at most $\Aa(Q,C,D)$ guarantees a solution that is no worse than $(m-1)!(n-1)!$ solutions. Thus, the domination ratio \cite{glover1997travelling,custic2017average} of such an algorithm is $\frac{1}{mn}$.

% define what it means to be arbitrarily worse (than optimal, average)?

%%%%%%%%%%%%%%%%%%%%%%%%%%%%%%%%%%%%%%%%%%%%%%%%%%%%%%%%%%%%%%%%
%%%%%%%%%%%%%%%%%%%%%%%%%%%%%%%%%%%%%%%%%%%%%%%%%%%%%%%%%%%%%%%%
%%%%%%%%%%%%%%%%%%%%%%%%%%%%%%%%%%%%%%%%%%%%%%%%%%%%%%%%%%%%%%%%
%%%%%%%%%%%%%%%%%%%%%%%%%%%%%%%%%%%%%%%%%%%%%%%%%%%%%%%%%%%%%%%%
\section{Construction heuristics}
\label{sec:constr}

In this section, we consider heuristic algorithms that will generate solutions to BAP using various construction approaches. Such algorithms are useful in situations where solutions of reasonable quality are needed  quickly. These algorithms can also be used to generate starting solutions for more complex improvement based algorithms.

\medskip

Our first algorithm, called \textit{\textbf{Random}}, is the trivial approach of generating a feasible  solution ($\mb{x},\mb{y}$). Both $\mb{x}$ and $\mb{y}$ are selected as random assignments in uniform fashion. It should be noted that the expected value of the solution produced by \textit{Random} is precisely $\Aa(Q,C,D)$.

\smallskip

Let us now discuss a different randomized technique, called \textit{\textbf{RandomXYGreedy}}. This algorithm builds a solution by randomly picking a `not yet assigned' $i \in M$ or $k \in N$, and then setting  $x_{ij}$ or $y_{kl}$ to $1$ for a `not yet assigned' $j \in M'$ or $l \in N'$ so that the total cost of the resulting \textit{partial solution} is minimized. A pseudo-code of \textit{RandomXYGreedy} is presented in Algorithm \ref{rxyg}. Here and later in the paper we will present description of the algorithms by assuming that the input BAP instance $(Q,C,D)$ has $C$ and $D$ as zero arrays. This restriction is for simplicity of presentation and does not affect neither the theoretical complexity of BAP nor the asymptotic computational complexity of the presented algorithms. It is easy to extend the algorithms to the general case in a straightforward way. The running time of \textit{RandomXYGreedy} is $O(mn^2)$ as each addition to our solution is selected using quadratic number of computations. However, just reading the data for the $Q$ matrix takes $O(m^2 n^2)$ time. For the rest of the paper we will consider running time of our algorithms without including this input overhead.

\begin{algorithm}

\caption{\textit{RandomXYGreedy}}
\label{rxyg}
\begin{algorithmic}[0]\scriptsize
\Input integers $m, n$; $m \times m \times n \times n$ array $Q$
\Output feasible solution to BAP

\State $x_{ij} \gets 0 \, \forall i, j$; $y_{kl} \gets 0 \, \forall k, l$
\While{not all $i \in M$ and $k \in N$ are assigned}
\State randomly pick some $i \in M$ or $k \in N$ that is unassigned

\If{$i$ is picked}
\State $j' \gets$ random $j \in M$ that is unassigned; $\Delta' \gets \sum_{k, l \in N} q_{ij'kl} y_{kl}$

\ForAll{$j \in M$ that is unassigned}
\State $\Delta \gets \sum_{k, l \in N} q_{ijkl} y_{kl}$ \Comment{value change if $i$ assigned to $j$}
\If{$\Delta < \Delta'$}
\State $j' \gets j$; $\Delta' \gets \Delta$
\EndIf
\EndFor

\State $x_{ij'} \gets 1$ \Comment{assign $i$ to $j'$}

\Else
\State $l' \gets$ random $l \in N$ that is unassigned; $\Delta' \gets \sum_{i, j \in M} q_{ijkl'} x_{ij}$

\ForAll{$l \in N$ that is unassigned}
\State $\Delta \gets \sum_{i, j \in M} q_{ijkl} x_{ij}$ \Comment{value change if $k$ assigned to $l$}
\If{$\Delta < \Delta'$}
\State $l' \gets l$; $\Delta' \gets \Delta$
\EndIf
\EndFor

\State $y_{kl'} \gets 1$ \Comment{assign $k$ to $l'$}

\EndIf
\EndWhile

\State \textbf{return} ($\mb{x}$, $\mb{y}$)
\end{algorithmic}
\end{algorithm}

Our next algorithm is fully deterministic and is called \textit{\textbf{Greedy}} (see Algorithm \ref{greedy}). This is similar to \textit{RandomXYGreedy}, except that, at each iteration, we select the best available $x_{ij}$ or $y_{kl}$ to be added to the current partial solution. We start the algorithm by choosing the partial solution $x_{i'j'} = 1$ and $y_{k'l'} = 1$ where $i', j', k', l'$ correspond to a smallest element in the array $Q$. The total running time of this heuristic is $O(n^3)$, considering that the position of the smallest $q_{i'j'k'l'}$ is provided.

\begin{algorithm}
\caption{\textit{Greedy}}
\label{greedy}
\begin{algorithmic}[0]\scriptsize
\Input integers $m, n$; $m \times m \times n \times n$ array $Q$
\Output feasible solution to BAP

\State $x_{ij} \gets 0 \, \forall i, j$; $y_{kl} \gets 0 \, \forall k, l$
\State $i', j', k', l' \gets arg\,min_{i, j \in M, k, l \in N} q_{ijkl}$; $x_{i'j'} \gets 1$; $y_{k'l'} \gets 1$

\While{not all $i \in M$ and $k \in N$ are assigned}
\State $\Delta'_x \gets \infty$; $\Delta'_y \gets \infty$

\ForAll{$i \in M$ that is unassigned}
\ForAll{$j \in M$ that is unassigned}
\State $\Delta \gets \sum_{k, l \in N} q_{ijkl} y_{kl}$ \Comment{value change if $i$ assigned to $j$}
\If{$\Delta < \Delta'_x$}
\State $i' \gets i$; $j' \gets j$; $\Delta'_x \gets \Delta$
\EndIf
\EndFor
\EndFor

\ForAll{$k \in N$ that is unassigned}
\ForAll{$l \in N$ that is unassigned}
\State $\Delta \gets \sum_{i, j \in M} q_{ijkl} x_{ij}$ \Comment{value change if $k$ assigned to $l$}
\If{$\Delta < \Delta'_y$}
\State $k' \gets k$; $l' \gets l$; $\Delta'_y \gets \Delta$
\EndIf
\EndFor
\EndFor

\If{$\Delta'_x \leq \Delta'_y$}
\State $x_{i'j'} \gets 1$ \Comment{assign $i'$ to $j'$}
\Else
\State $y_{k'l'} \gets 1$ \Comment{assign $k'$ to $l'$}
\EndIf
\EndWhile

\State \textbf{return} ($\mb{x}$, $\mb{y}$)
\end{algorithmic}
\end{algorithm}

\begin{theorem}
The objective function value of a solution produced by the Greedy algorithm could be arbitrarily bad and could be worse than $\Aa(Q,C,D)$.
\end{theorem}
\begin{proof}
Consider the following BAP instance: $C$ and $D$ are zero matrices and elements of $2 \times 2 \times 3 \times 3$ matrix $Q$ are all zero except $q_{1111}=-\epsilon, q_{1122}=q_{1133}=\epsilon, q_{2211}=q_{1123}=q_{1132}=2\epsilon, q_{2222}=q_{2233}=L$, where $\epsilon$ and $L$ are arbitrarily small and large positive numbers, respectively. At first the algorithm will assign $x_{11} = y_{11} = 1$, as $q_{1111}$ is the smallest element in the array. Next, all indices $i, j \in M$ such that $i, j > 2$ and $k, l \in M$ such that $k, l > 3$ will be assigned within their respective groups. This is due to the fact that any assignment in those sets adds no additional cost to the current partial solution. Following that, $y_{22} = y_{33} = 1$ will be added. And finally, $x_{22}$ will be set to $1$ to complete a solution with the cost $3\epsilon + 2L$. However, an optimal solution in this case will contain $x_{11} = x_{22} = y_{11} = y_{23} = y_{32} = 1$ with an objective value of $5\epsilon$. Note that $\Aa(Q,C,D) = \frac{7\epsilon + 2L}{mn}$ and the result follows.
\end{proof}

We also consider a randomized version of \textit{Greedy}, called \textit{\textbf{GreedyRandomized}}. In this variation a partial assignment is extended by a randomly picked $x_{ij}$ or $y_{kl}$ out of $h$ best candidates (by solution value change), where $h$ is some fixed number. Such approaches are generally called semi-greedy algorithms and form an integral part of many GRASP algorithms \cite{hart1987semi,feo1989probabilistic}. To emphasize the randomized decisions in the algorithm and its linkages to GRASP, we call it \textit{GreedyRandomized}.

\medskip

Finally we discuss a construction heuristic based on rounding a fractional solution. In \cite{custicbilinear}, a discretization procedure was introduced that computes a feasible solution to BAP with objective function value no more than that of the fractional solution. Given a fractional solution to BAP ($\mb{x},\mb{y}$) (i.e. a solution to BAP (\ref{of})-(\ref{y2}) without integrality constrains (\ref{int})), we fix one side of the solution (say $\mb{x}$) and optimize $\mb{y}$ by solving a linear assignment problem to obtain a solution $\mb{\bar{y}}$ . Then, fix $\mb{\bar{y}}$ and solve a linear assignment problem to find a solution $\mb{\bar{x}}$. Output the solution ($\mb{\bar{x}},\mb{\bar{y}}$) as a result. We denote this approach as \textit{\textbf{Rounding}}.

\begin{theorem}
A feasible solution $(\mb{x}^*,\mb{y}^*)$ to BAP with the cost $f(\mb{x}^*,\mb{y}^*)\leq \Aa(Q,C,D)$, can be obtained in $O(m^2n^2+n^3)$ time using the \textit{Rounding} algorithm.
\end{theorem}
\begin{proof}
Consider the fractional solution $(\mb{x},\mb{y})$ where $x_{ij}=1/m$ for all $i,j\in M$, and $y_{ij}=1/n$ for all $i,j\in N$. Then $(\mb{x},\mb{y})$ is a feasible solution to the relaxation of BAP obtained by removing the integrality restrictions \eqref{int}. It is easy to see that $f(\mb{x},\mb{y})=\Aa(Q,C,D)$. One of the properties of \textit{Rounding} discussed in \cite{custicbilinear} is that the resulting solution is no worse than the input fractional solution, in terms of objective value. Apply Rounding to $(\mb{x},\mb{y})$ to obtain the desired solution.
\end{proof}

\textit{Rounding} provides us with an alternative way to Theorem \ref{thm:minmax} for generating a BAP solution with objective value no worse than the average. Recall, that by Theorem \ref{thm:dom} this solution is guaranteed to be no worse than $(m-1)!(n-1)!$ feasible solutions.

It should be noted that this discretization procedure could also be applied to BAP fractional solutions obtained from other sources, such as the solution to the relaxed version of an integer linear programming reformulation of BAP. Some of the linearization reformulations \cite{kaufman1978algorithm,frieze1989algorithms,lawler1963quadratic,adams2007level} of the QAP can be modified to obtain the corresponding linearizations of BAP. Selecting only $\mb{x}$ and $\mb{y}$ part from continuous solutions and ignoring other variables in the linearization formulations can be used to initiate the rounding algorithm discussed above. However, in this case, the resulting solution is not guaranteed to be no worse than the average.

%%%%%%%%%%%%%%%%%%%%%%%%%%%%%%%%%%%%%%%%%%%%%%%%%%%%%%%%%%%%%%%%
%%%%%%%%%%%%%%%%%%%%%%%%%%%%%%%%%%%%%%%%%%%%%%%%%%%%%%%%%%%%%%%%
%%%%%%%%%%%%%%%%%%%%%%%%%%%%%%%%%%%%%%%%%%%%%%%%%%%%%%%%%%%%%%%%
%%%%%%%%%%%%%%%%%%%%%%%%%%%%%%%%%%%%%%%%%%%%%%%%%%%%%%%%%%%%%%%%
\section{Neighborhood structures and properties}
\label{sec:ls}

Let us now discuss various neighborhoods associated with a feasible solution of BAP and analyze their properties. We also consider worst case properties of a local optimum for these neighborhoods. All these neighborhoods are based on reassigning parts of $\mb{x}\in \xx$, parts of $\mb{y}\in\yy$, or both. The neighborhoods that we consider can be classified into three categories: \textit{$h$-exchange neighborhoods}, \textit{$[h,p]$-exchange neighborhoods}, and \textit{shift based neighborhoods}.

%%%%%%%%%%%%%%%%%%%%%%%%%%%%%%%%%%%%%%%%%%%%%%%%%%%%%%%%%%%%%%%%
%%%%%%%%%%%%%%%%%%%%%%%%%%%%%%%%%%%%%%%%%%%%%%%%%%%%%%%%%%%%%%%%
\subsection{The $h$-exchange neighborhood}
\label{sec:hex}

In this class of neighborhoods, we apply an $h$-exchange operation to $ \mb{x}$ while keeping $\mb{y}$ unchanged or viceversa. Let us discuss this in detail with $h = 2$. The $2$-exchange neighborhood is well studied in the QAP literature. Our version of $2$-exchange for BAP is related to the QAP variation, but also have some significant differences due to the specific structure of our problem.

\medskip

Let $(\mb{x}, \mb{y})$ be a feasible solution to BAP. Consider two elements $i_1, i_2 \in M$, $j_1, j_2 \in M'$, such that $x_{i_1j_1} = x_{i_2j_2} = 1$. Then the \textit{$2$-exchange} operation on the $\mb{x}$-variables produces $(\mb{x}', \mb{y})$, where $\mb{x}'$ is obtained from $\mb{x}$ by swapping assignments of $i_1, i_2$ and $j_1, j_2$ (i.e. setting $x_{i_1j_2} = x_{i_2j_1} = 1$ and $x_{i_1j_1} = x_{i_2j_2} = 0$). Let $\Delta^x_{i_1i_2}$ be the change in the objective value from $(\mb{x}, \mb{y})$ to $(\mb{x}', \mb{y})$. I.e.,
\begin{equation}
\begin{split}
\Delta^x_{i_1i_2} = & f(\mb{x}', \mb{y}) - f(\mb{x}, \mb{y})\\
= & \sum_{i=1}^m\sum_{j=1}^m\sum_{k=1}^n\sum_{l=1}^n q_{ijkl}x'_{ij}y_{kl} + \sum_{i=1}^m\sum_{j=1}^m c_{ij}x'_{ij} + \sum_{k=1}^n\sum_{l=1}^n d_{kl}y_{kl}\\
& - \sum_{i=1}^m\sum_{j=1}^m\sum_{k=1}^n\sum_{l=1}^n q_{ijkl}x_{ij}y_{kl} - \sum_{i=1}^m\sum_{j=1}^m c_{ij}x_{ij} - \sum_{k=1}^n\sum_{l=1}^n d_{kl}y_{kl}\\
= & \sum_{k=1}^n\sum_{l=1}^n (q_{i_1j_2kl} + q_{i_2j_1kl} - q_{i_1j_1kl} - q_{i_2j_2kl}) y_{kl} + c_{i_1j_2} + c_{i_2j_1} - c_{i_1j_1} - c_{i_2j_2}.
\end{split}
\end{equation}

Let $2exchangeX(\mb{x}, \mb{y})$ be the set of all feasible solutions $(\mb{x}', \mb{y})$, obtained from $(\mb{x}, \mb{y})$ by applying the $2$-exchange operation for all $i_1, i_2 \in M$ (with corresponding $j_1, j_2 \in M'$). Efficient computation of $\Delta^x_{i_1i_2}$ is crucial in developing fast algorithms that use this neighborhood.   For a fixed $\mb{y}$, consider the $m \times m$ matrix $E$ such that $e_{ij} = \sum_{k=1}^n\sum_{l=1}^n q_{ijkl} y_{kl} + c_{ij}$. Then we can write $\Delta^x_{i_1i_2} = e_{i_1j_2} + e_{i_2j_1} - e_{i_1j_1} - e_{i_2j_2}$. If the matrix $E$ is available, this calculation can be done in constant time, and hence the neighborhood $2exchangeX(\mb{x}, \mb{y})$ can be explored in $O(m^2)$ time for an improving solution.  Note that the values of $E$ depend only on $\mb{y}$ and not on $\mb{x}$. Thus, we do not need to update $E$ within a local search algorithm as long as $\mb{y}$ remains unchanged.

Likewise, we can define a 2-exchange operation on $\mb{y}$ by keeping $\mb{x}$ constant. Consider two elements $k_1, k_2 \in N$ and let $l_1, l_2$ be the corresponding assignments in $N'$, such that $x_{k_1l_1} = x_{k_2l_2} = 1$. Then the $2$-exchange operation will  produce $(\mb{x}, \mb{y}')$, where $\mb{y}'$ is obtained from $\mb{y}$ by swapping assignments of $k_1, k_2$ and $l_1, l_2$ (i.e. setting $x_{k_1l_2} = x_{k_2l_1} = 1$ and $x_{k_1l_1} = x_{k_2l_2} = 0$). Let $\Delta^y_{k_1k_2}$ be the change in the objective value from $(\mb{x}, \mb{y})$ to $(\mb{x}, \mb{y}')$. I.e.,

\begin{equation}
\begin{split}
\Delta^y_{k_1k_2} = & f(\mb{x}, \mb{y}') - f(\mb{x}, \mb{y})\\
= & \sum_{i=1}^m\sum_{j=1}^m (q_{ijk_1l_2} + q_{ijk_2l_1} - q_{ijk_1l_1} - q_{ijk_2l_2}) x_{ij} + d_{k_1l_2} + d_{k_2l_1} - d_{k_1l_1} - d_{k_2l_2}.
\end{split}
\end{equation}

Let $2exchangeY(\mb{x}, \mb{y})$ be the set of all feasible solutions $(\mb{x}, \mb{y}')$, obtained from $(\mb{x}, \mb{y})$ by applying the 2-exchange operation on $\mb{y}$ while keeping $\mb{x}$ unchanged. As in the previous case, efficient computation of $\Delta^y_{k_1k_2}$ is crucial in developing fast algorithms that use this neighborhood.   For a fixed $\mb{x}$ consider an $n \times n$ matrix $G$ such that $g_{kl} = \sum_{i=1}^m\sum_{j=1}^m q_{ijkl} x_{ij} + d_{kl}$. Then we can write $\Delta^y_{k_1k_2} = g_{k_1l_2} + g_{k_2l_1} - g_{k_1l_1} - g_{k_2l_2}$. If the matrix $G$ is available, this calculation can be done in constant time and hence the neighborhood $2exchangeY(\mb{x}, \mb{y})$ can be explored in $O(n^2)$ time for an improving solution.  Note that the values of $G$ depends only on $\mb{x}$ and not on $\mb{y}$. Thus, we do not need to update $G$ within a local search algorithm as long as $\mb{y}$ remains unchanged.

The \textit{2-exchange neighborhood} of ($\mb{x}, \mb{y}$), denoted by $2exchange(\mb{x}, \mb{y})$, is given by $$2exchange(\mb{x}, \mb{y}) = 2exchangeX(\mb{x}, \mb{y})\cup 2exchangeY(\mb{x}, \mb{y}).$$

In a local search algorithm based on the $2exchange(\mb{x}, \mb{y})$ neighborhood, after each move, either $\mb{x}$ or $\mb{y}$ will be changed, but not both. To maintain our data structure, if $\mb{y}$ is changed, we update $E$ in $O(m^2)$ time. More specifically, suppose a $2$-exchange operation takes $(\mb{x}, \mb{y})$ to $(\mb{x}, \mb{y}')$, then $E$ is updated as: $e_{ij} \gets e_{ij} + q_{ijk_1l_2} + q_{ijk_2l_1} - q_{ijk_1l_1} - q_{ijk_2l_2}$, where $k_1, k_2 \in N, l_1, l_2 \in N'$ are the corresponding positions where the swap have occurred. Analogous changes will be performed on $G$ in $O(n^2)$ time if $(\mb{x}, \mb{y})$ is changed to $(\mb{x}', \mb{y})$.

\smallskip

The general \textit{$h$-exchange neighborhood} for BAP is obtained by replacing $2$ in the above definition by $2, 3, \ldots, h$. Notice that the $h$-exchange neighborhood can be searched for an improving solution in $O(n^h)$ time, and already for $h=3$, the running time of the algorithm that completely explores this neighborhood is $O(n^3)$. With the same asymptotic running time we could instead optimally reassign whole $\mb{x}$ (or $\mb{y}$) by solving the linear assignment problem with $E$ (or $G$ respectively) as the cost matrix. This fact suggests that any $h$ larger that $3$ potentially leads to a weaker algorithm in terms of running time.
Such full reassignment can be viewed as a local search based on the special case of the $h$-exchange neighborhood with $h = n$. This special local search will be referred to as \textit{\textbf{Alternating Algorithm}} and will be alternating between re-optimizing $\mb{x}$ and $\mb{y}$. For clarity, the pseudo code for this approach is presented in Algorithm \ref{AA}. \textit{Alternating Algorithm} is a strategy well-known in non-linear programming literature as \textit{coordinate-wise descent}. Similar underlying ideas are used in the context of other bilinear programming problems by various authors \cite{konno1980maximizing,karapetyan2012heuristic,punnen2015average}.

\begin{algorithm}
\caption{\textit{Alternating Algorithm}}
\label{AA}
\begin{algorithmic}[0]\scriptsize
\Input integers $m, n$; $m \times m \times n \times n$ array $Q$; feasible solution ($\mb{x}$, $\mb{y}$) to  BAP
\Output feasible solution to BAP

\While{True}

\State $e_{ij} \gets \sum_{k, l \in N} q_{ijkl} y_{kl} \, \forall i, j \in M$
\State $\mb{x}^* \gets arg\,min_{\mb{x}' \in \xx} \sum_{i, j \in M} e_{ij} x'_{ij}$ \Comment{solving assignment problem for $\mb{x}$}

\State $g_{kl} \gets \sum_{i, j \in M} q_{ijkl} x^*_{ij} \, \forall k, l \in N$
\State $\mb{y}^* \gets arg\,min_{\mb{y}' \in \yy} \sum_{k, l \in N} g_{kl} y'_{kl}$ \Comment{solving assignment problem for $\mb{y}$}

\If{$f(\mb{x}^*, \mb{y}^*) = f(\mb{x}, \mb{y})$ }
\State \textbf{break}
\EndIf

\State $\mb{x} \gets \mb{x}^*; \, \mb{y} \gets \mb{y}^*$

\EndWhile
\State \textbf{return} ($\mb{x}$, $\mb{y}$)
\end{algorithmic}
\end{algorithm}

\begin{theorem}
\label{thm:hex}
The objective function value of a locally optimal solution for BAP based on the h-exchange neighborhood could be arbitrarily bad and could be worse than $\Aa(Q,C,D)$, for any $h$.
\end{theorem}
\begin{proof}
For a small $\epsilon > 0$ and a large $L$, we consider BAP instance  $(Q,C,D)$ such that all of its cost elements are equal to $0$, except $c_{11}=c_{22}=d_{11}=d_{22}=-\epsilon$, and $q_{1212}=-L$. Let a feasible solution $(\mb{x},\mb{y})$ be such that $x_{11}=x_{22}=y_{11}=y_{22}=1$. Then $(\mb{x},\mb{y})$ is a local optimum for the $h$-exchange neighborhood. Note that this local optimum can only be improved by simultaneously making changes to both $\mb{x}$ and $\mb{y}$, which is not possible for this neighborhood. The objective function value of $(\mb{x},\mb{y})$ is $-4\epsilon$, while the optimal solution objective value is $-L$.
\end{proof}

Despite the negative result of Theorem \ref{thm:hex}, we will see in Section \ref{sec:explsms} that on average, $2$-exchange and $n$-exchange (with \textit{Alternating Algorithm}) are two of the most efficient neighborhoods to explore from a practical point of view.
Moreover, when restricted to non-negative input array, we can establish some performance guarantees for $2$-exchange (and consequently for any $h$-exchange) local search. In particular, we derive upper bounds on the local optimum solution value and the number of iterations to reach a solution not worse than this value bound. The proof technique follows \cite{angel1998quality}, where authors obtained similar bounds for Koopmans-Beckman QAP.
In fact, these results can be obtained for the general QAP as well, by modifying the following proof accordingly.

\begin{theorem}
\label{thm:2exA}
For any BAP instance $(Q,C,D)$ with non-negative $Q$ and zero matrices $C,D$, the cost of the local optimum for the $2$-exchange neighborhood is $f^* \leq \frac{2mn}{m + n} \Aa(Q,C,D)$.
\end{theorem}
\begin{proof}
In this proof, for simplicity, we represent BAP as a permutation problem. As such, the permutation formulation of BAP is
\begin{equation}
\min_{\pi \in \Pi, \phi \in \Phi} \sum_{i=1}^m \sum_{k=1}^n q_{i \, \pi(i) \, k \, \phi(k)},
\end{equation}
where $\Pi$ and $\Phi$ are sets of all permutations on $\{1, 2, \ldots, m\}$ and $\{1, 2, \ldots, n\}$, respectively. Cost of a particular permutation pair $\pi, \phi$ is $f(\pi, \phi) = \sum_{i=1}^m \sum_{k=1}^n q_{i \, \pi(i) \, k \, \phi(k)}$.

Let $\pi_{ij}$ be the permutation obtained by applying a single $2$-exchange operation to $\pi$ on indices $i$ and $j$. Define $\delta^{\pi}_{ij}$ as an objective value difference after applying such $2$-exchange:
\begin{equation}
\delta^{\pi}_{ij}(\pi, \phi) = f(\pi_{ij}, \phi) - f(\pi, \phi) = \sum_{k = 1}^m \left(q_{i \, \pi(j) \, k \, \phi(k)} + q_{j \, \pi(i) \, k \, \phi(k)} - q_{i \, \pi(i) \, k \, \phi(k)} - q_{j \, \pi(j) \, k \, \phi(k)}\right). \nonumber
\end{equation}
Similarly we can have $\phi_{kl}$ and $\delta^{\phi}_{kl}$:
\begin{equation}
\delta^{\phi}_{kl}(\pi, \phi) = f(\pi, \phi_{kl}) - f(\pi, \phi) = \sum_{i = 1}^n \left(q_{i \, \pi(i) \, k \, \phi(l)} + q_{i \, \pi(i) \, l \, \phi(k)} - q_{i \, \pi(i) \, k \, \phi(k)} - q_{i \, \pi(i) \, l \, \phi(l)}\right). \nonumber
\end{equation}
Summing up over all possible $\delta^{\pi}_{ij}$ and $\delta^{\phi}_{kl}$ we get
\begin{align}
\label{deltapi}
\sum_{i,j=1}^m \delta^{\pi}_{ij}(\pi, \phi) \, &= \sum_{i,j=1}^m \sum_{k=1}^n q_{i \, \pi(j) \, k \, \phi(k)} + \sum_{i,j=1}^m \sum_{k=1}^n q_{j \, \pi(i) \, k \, \phi(k)} - \sum_{i,j=1}^m \sum_{k=1}^n q_{i \, \pi(i) \, k \, \phi(k)} - \sum_{i,j=1}^m \sum_{k=1}^n q_{j \, \pi(j) \, k \, \phi(k)} \nonumber\\
&= 2 \sum_{i,j=1}^m \sum_{k=1}^n q_{i \, \pi(j) \, k \, \phi(k)} - 2m f(\pi, \phi),
\end{align}
\begin{align}
\label{deltaphi}
\sum_{k,l=1}^n \delta^{\phi}_{kl}(\pi, \phi) = 2 \sum_{i=1}^m \sum_{k,l=1}^n q_{i \, \pi(i) \, k \, \phi(l)} - 2n f(\pi, \phi).
\end{align}
Using (\ref{deltapi}) and (\ref{deltaphi}) we can now compute an average cost change after $2$-exchange operation on solution $(\pi, \phi)$.
\begin{align}
\label{DeltaB}
\Delta(\pi, \phi) \, &= \frac{\sum_{i,j=1}^m \delta^{\pi}_{ij}(\pi, \phi) + \sum_{k,l=1}^n \delta^{\phi}_{kl}(\pi, \phi)}{m^2 + n^2} \nonumber\\
&= \frac{2 \sum_{i,j=1}^m \sum_{k=1}^n q_{i \, \pi(j) \, k \, \phi(k)} + 2 \sum_{i=1}^m \sum_{k,l=1}^n q_{i \, \pi(i) \, k \, \phi(l)} - 2(m + n)f(\pi, \phi)}{m^2 + n^2} \nonumber\\
& = \frac{2 \sum_{i,j=1}^m \sum_{k=1}^n q_{i \, \pi(j) \, k \, \phi(k)} + 2 \sum_{i=1}^m \sum_{k,l=1}^n q_{i \, \pi(i) \, k \, \phi(l)}}{m^2 + n^2} \! - \! \lambda f(\pi, \phi) \! + \! \lambda \frac{2mn}{m \! + \! n} \Aa \! - \! \lambda \frac{2mn}{m \! + \! n} \Aa \nonumber\\
& \leq - \lambda(f(\pi, \phi) - \frac{2mn}{m + n} \Aa) + \mu - \lambda \frac{2mn}{m + n} \Aa,
\end{align}
where $\lambda = 2 \frac{m + n}{m^2 + n^2}$ and $\mu = \max_{\pi \in \Pi, \phi \in \Phi} \left[\dfrac{2 \sum_{i,j=1}^m \sum_{k=1}^n q_{i \, \pi(j) \, k \, \phi(k)} + 2 \sum_{i=1}^m \sum_{k,l=1}^n q_{i \, \pi(i) \, k \, \phi(l)}}{m^2 + n^2}\right]$. Note that both $\lambda$ and $\mu$ do not depend on any particular solution and are fixed for a given BAP instance.

We are ready to prove the theorem by contradiction. Let $(\pi^*, \phi^*)$ be the local optimum for $2$-exchange local search, with the objective function cost $f^* = f(\pi^*, \phi^*)$. Assume now that $f(\pi^*, \phi^*) > \frac{2mn}{m + n} \Aa$. Then $-\lambda(f(\pi^*, \phi^*) - \frac{2mn}{m + n} \Aa) < 0$ and

\begin{align}
\mu - \lambda \frac{2mn}{m + n} \Aa = \max_{\pi \in \Pi, \phi \in \Phi} \, &\left[\frac{2 \sum_{i,j=1}^m \sum_{k=1}^n q_{i \, \pi(j) \, k \, \phi(k)} + 2 \sum_{i=1}^m \sum_{k,l=1}^n q_{i \, \pi(i) \, k \, \phi(l)}}{m^2 + n^2}\right] \nonumber\\
- \, &2 \frac{m + n}{m^2 + n^2} \frac{2mn}{m+n} \frac{1}{mn} \sum_{i,j=1}^m \sum_{k,l=1}^n q_{ijkl} \nonumber\\
= \max_{\pi \in \Pi, \phi \in \Phi} \, &\left[\frac{2 \sum_{i,j=1}^m \sum_{k=1}^n q_{i \, \pi(j) \, k \, \phi(k)}}{m^2 + n^2} + \frac{2 \sum_{i=1}^m \sum_{k,l=1}^n q_{i \, \pi(i) \, k \, \phi(l)}}{m^2 + n^2}\right] \nonumber\\
- \, &\frac{2 \sum_{i,j=1}^m \sum_{k,l=1}^n q_{ijkl}}{m^2 + n^2} - \frac{2 \sum_{i,j=1}^m \sum_{k,l=1}^n q_{ijkl}}{m^2 + n^2} \leq 0,
\end{align}
which implies $\Delta(\pi^*, \phi^*) < 0$. As $\Delta$ is the average cost difference after applying $2$-exchange, there exists some swap that decreases solution cost by at least $-\Delta(\pi^*, \phi^*)$, and that contradicts with $(\pi^*, \phi^*)$ being a local optimum.
\end{proof}

It is easy to see that the bound $\mu \leq \lambda \frac{2mn}{m + n} \Aa$ from Theorem \ref{thm:2exA} is tight. Consider some arbitrary bilinear assignment $(\pi, \phi)$, and set all $q_{ijkl}$ to zero except $q_{i \, \pi(i) \, k \, \phi(k)} = 1, \, \forall i \, \forall k$. Then $\mu = 4 \dfrac{\sum_{i=1}^m \sum_{k=1}^n q_{i \, \pi(i) \, k \, \phi(k)}}{m^2 + n^2} = \lambda \frac{2mn}{m + n} \Aa = \frac{4mn}{m^2 + n^2}$.

\begin{theorem}
\label{thm:2ext}
For any BAP instance $(Q,C,D)$ with elements of $Q$ restricted to non-negative integers and zero matrices $C,D$, the local search algorithm that explores $2$-exchange neighborhood will reach a solution with the cost at most $\frac{2mn}{m + n} \Aa(Q,C,D)$ in $O\left(\frac{m^2 + n^2}{m + n} \log{\sum q_{ijkl}}\right)$ iterations.
\end{theorem}
\begin{proof}
Inequality (\ref{DeltaB}) can be also written as $\Delta(\pi, \phi) \leq -\lambda f(\pi, \phi) + \mu$, and so any solution with $f(\pi, \phi) > \frac{\mu}{\lambda}$ would yield $\Delta(\pi, \phi) < 0$, and would have some $2$-exchange improvement possible. Note that $\frac{2mn}{m + n} \Aa \geq \frac{\mu}{\lambda}$.

Consider a cost $f'(\pi, \phi) = f(\pi, \phi) - \frac{\mu}{\lambda}$. At every step of the $2$-exchange local search $f'(\pi, \phi)$ is decreased by at least $\Delta(\pi, \phi)$ and becomes at most $$f'(\pi, \phi) + \Delta(\pi, \phi) \leq f'(\pi, \phi) + (-\lambda f(\pi, \phi) + \mu) = f'(\pi, \phi) - \lambda f'(\pi, \phi) = (1 - \lambda)f'(\pi, \phi).$$ Since elements of $Q$ are integer, the cost at each step must decrease by at least $1$. Then a number of iterations $t$ for $C'(\pi, \phi)$ to become less than or equal to zero has to satisfy 
\begin{align}
(1 - \lambda)^{t-1} (f_{\text{max}} - \frac{\mu}{\lambda}) - (1 - \lambda)^{t} (f_{\text{max}} - \frac{\mu}{\lambda}) &\geq 1, \nonumber\\
(1 - \lambda)^{t-1} (f_{\text{max}} - \frac{\mu}{\lambda}) (1 - (1 - \lambda)) &\geq 1, \nonumber\\
(1 - \lambda)^{t-1} &\geq \frac{1}{(f_{\text{max}} - \frac{\mu}{\lambda}) \lambda}, \nonumber\\
(t - 1) \log{(1 - \lambda)} &\geq -\log{\lambda (f_{\text{max}} - \frac{\mu}{\lambda})}, \nonumber\\
t &\leq 1 + \frac{-\log{\lambda (f_{\text{max}} - \frac{\mu}{\lambda}})}{\log{(1 - \lambda)}},
\end{align}
where $f_{\text{max}}$ is the highest possible solution value. It follows that
\begin{equation}
t \in O\left(\frac{1}{\lambda} \log{\lambda (f_{\text{max}} - \frac{\mu}{\lambda}})\right) = O\left(\frac{m^2 + n^2}{m + n} \log{\frac{m + n}{m^2 + n^2} (f_{\text{max}} - \frac{\mu}{\lambda}})\right).
\end{equation}
This together with the fact that $f_{\text{max}} - \frac{\mu}{\lambda} \leq f_{\text{max}} \leq \sum_{i,j=1}^m \sum_{k,l=1}^n q_{ijkl}$ completes the proof.
\end{proof}

It should be noted that the solution considered in the statement of Theorem \ref{thm:2ext} may not be a local optimum. The theorem simply states that, the solution of the desired quality will be reached by $2$-exchange local search in polynomial time. It is known that for QAP, $2$-exchange local search may sometimes reach local optimum in exponential number of steps \cite{pardalos1994quadratic}.

%%%%%%%%%%%%%%%%%%%%%%%%%%%%%%%%%%%%%%%%%%%%%%%%%%%%%%%%%%%%%%%%
%%%%%%%%%%%%%%%%%%%%%%%%%%%%%%%%%%%%%%%%%%%%%%%%%%%%%%%%%%%%%%%%
\subsection{[$h,p$]-exchange neighborhoods}
\label{sec:hpex}

Recall that in the $h$-exchange neighborhood we change either the $\mb{x}$ variables or the $\mb{y}$ variables, but not both. Simultaneous changes in $\mb{x}$ and $\mb{y}$ could lead to more powerful neighborhoods, but with additional computational effort in exploring them. With this motivation, we introduce the \textit{[$h,p$]-exchange neighborhood }for BAP.

\medskip

In the $[h,p]$-exchange neighborhood, for each $h$-exchange operation on $\mb{x}$ variables, we consider all possible $p$-exchange operations on $\mb{y}$ variables. Thus, the [$h,p$]-exchange neighborhood is the set of all solutions $(\mb{x}', \mb{y}')$ obtained from the given solution $(\mb{x}, \mb{y})$, such that $\mb{x}'$ differs from $\mb{x}$ in at most $h$ assignments, and $\mb{y}'$ differs from $\mb{y}$ in at most $p$ assignments. The size of this neighborhood is $\Theta(m^hn^p)$.

\begin{theorem}
The objective function value of a locally optimal solution for the $[h,p]$-exchange neighborhood could be arbitrarily bad. If $h < \frac{m}{2}$ or $p < \frac{n}{2}$ this value could be arbitrarily worse than $\Aa(Q,C,D)$.
\end{theorem}
\begin{proof}
Let $\epsilon > 0$ be an arbitrarily small and $L$ be an arbitrarily large numbers. Consider the BAP instance $(Q,C,D)$ such that all of the associated cost elements are equal to $0$, except $q_{iikk}=-\epsilon, q_{i(i+1 \mod m)k(k+1 \mod n)}=-L, q_{iik(k+1 \mod n)}=\frac{hL}{m - h} \quad \forall i \in M \, \forall k \in N$. Let $(\mb{x},\mb{y})$ be a feasible solution such that $x_{ii}=1 \quad \forall i \in M$ and $y_{kk}=1 \quad \forall k \in N$. Note that $f(\mb{x},\mb{y}) = -mn\epsilon$.

We first show that $(\mb{x},\mb{y})$ is a local optimum for the $[h,p]$-exchange neighborhood. If we assume the opposite and $(\mb{x},\mb{y})$ is not a local optimum, then there exist a solution $(\mb{x}',\mb{y}')$ with $\mb{x}'$ being different from $\mb{x}$ in at most $h$ assignments, $\mb{y}'$ being different from $\mb{y}$ in at most $p$ assignments, and $f(\mb{x}',\mb{y}') - f(\mb{x},\mb{y}) < 0$. Since the summation for $f(\mb{x},\mb{y})$ comprised of exactly $mn$ elements of $Q$ with value $-\epsilon$, the only way to get an improving solution is to get some number of elements with value $-L$, and therefore to flip some number of $x_{ii}$ to $x_{i(i + 1 \mod m)}$ and $y_{kk}$ to $y_{k(k + 1 \mod n)}$. Let $1 < u \leq h$ and $1 < v \leq p$ be the number of such elements $u = |\{i \in M | x'_{i(i + 1 \mod m)} = 1\}|$ and $v = |\{k \in N | y'_{k(k + 1 \mod n)} = 1\}|$ in $(\mb{x}',\mb{y}')$. Then we know that the cost function $f(\mb{x}',\mb{y}')$ contains exactly $uv$ number of $-L$. However, each of the $v$ elements of type $y'_{k(k + 1 \mod n)} = 1$ also contributes at least $(m - h) \frac{hL}{m - h} = hL$ to the objective value (due to remaining $m - h$ elements of type $x_{ii} = 1$ being unchanged). From this we get that $f(\mb{x}',\mb{y}') > mn(-\epsilon) + uv(-L) + hv(L) = f(\mb{x},\mb{y}) + vL(h - u)$, and since $u \leq h$ we get $f(\mb{x}',\mb{y}') - f(\mb{x},\mb{y}) > 0$ which contradicts the fact that $(\mb{x}',\mb{y}')$ is an improving solution to $(\mb{x},\mb{y})$. Hence, $(\mb{x},\mb{y})$ must be a local optimum.

We also get that an optimal solution for this instance is $x_{i(i+1 \mod m)}=1 \quad \forall i \in M$ and $y_{k(k+1 \mod n)}=1 \quad \forall k \in N$ with a total cost of $-mnL$. The average value of all feasible solutions is $\Aa(Q,C,D) = \dfrac{mn(-L) + mn(-\epsilon) + mn\frac{hL}{m - h}}{mn} = L\frac{2h-m}{m - h} - \epsilon$. $h < \frac{m}{2}$ and appropriate choice of $\epsilon, L$ guarantee us that considered local optimum is arbitrarily worse than $\Aa(Q,C,D)$. The construction of the example for the case $p < \frac{n}{2}$ is similar, so we omit the details.
\end{proof}

\smallskip

One particular case of the $[h,p]$-exchange neighborhood deserves a special mention. If $p=n$, then for each candidate $h$-exchange solution $\mb{x}'$ we will consider all possible assignments for $\mb{y}$. To find the optimal $\mb{y}$ given $\mb{x}'$, we can solve a linear assignment problem with cost matrix $g_{kl} = \sum_{i=1}^m\sum_{j=1}^m q_{ijkl} x'_{ij} + d_{kl}$, as in the \textit{Alternating Algorithm}. Analogous situation appears when we consider $[h,p]$-exchange neighborhood with $h=m$.

A set of solutions defined by the union of $[h,n]$-exchange and $[m,p]$-exchange neighborhoods, for the case $h=p$, will be called simply \textit{optimized $h$-exchange neighborhood}.
Note that the optimized $h$-exchange neighborhood is exponential in size, but it can be searched in $O(m^hn^3 + n^hm^3)$ time due to the fact that for fixed $\mb{x}$ ($\mb{y}$), optimal $f(\mb{x}, \mb{y}')$ ($f(\mb{x}', \mb{y})$) can be found in $O(n^3)$ time. Neighborhoods similar to optimized $2$-exchange were used for unconstrained bipartite binary quadratic program by Glover et al. \cite{glover2015integrating}, and for the bipartite quadratic assignment problem by Punnen and Wang \cite{punnen2016bipartite}.

As in the case of $h$-exchange, some performance bounds for optimized $h$-exchange neighborhood can be established, if the input array $Q$ is not allowed to have negative elements.

\begin{theorem}
\label{thm:2exOptA}
There exists a solution with the cost $f \leq (m + n) \Aa(Q,C,D)$ in the optimized $2$-exchange neighborhood of every solution to BAP, for any instance $(Q,C,D)$ with non-negative $Q$ and zero matrices $C,D$.
\end{theorem}
\begin{proof}
The proof will follow the structure of Theorem \ref{thm:2exA}, and will focus on the average solution change to a given permutation pair solution ($\pi$, $\phi$) to BAP.

Let $\pi_{ij}$ be the permutation obtained by applying a single $2$-exchange operation to $\pi$ on indices $i$ and $j$, and $\phi^*$ be the optimal permutation that minimizes the solution cost for such fixed $\pi_{ij}$. Define $\delta^{\pi}_{ij}$ as the objective value difference after applying such operation:
\begin{equation}
\delta^{\pi}_{ij}(\pi, \phi) = f(\pi_{ij}, \phi^*) - f(\pi, \phi) = \sum_{u=1}^m \sum_{k=1}^n q_{u \, \pi_{ij}(u) \, k \, \phi^*(k)} - f(\pi, \phi) \leq \frac{1}{n} \sum_{u = 1}^m \sum_{k,l = 1}^n q_{u \, \pi_{ij}(u) \, k \, l} - f(\pi, \phi). \nonumber
\end{equation}
The last inequality due to the fact that, for fixed $\pi_{ij}$, the value of the solution with the optimal $\phi^*$ is not worse than the average value of all such solutions.
We also know that for any $k,l \in N$,
\begin{equation}
\sum_{u = 1}^m q_{u \, \pi_{ij}(u) \, k \, l} = \sum_{u = 1}^m q_{u \, \pi(u) \, k \, l} + q_{i \, \pi(j) \, k \, l} + q_{j \, \pi(i) \, k \, l} - q_{i \, \pi(i) \, k \, l} - q_{j \, \pi(j) \, k \, l}. \nonumber
\end{equation}
and, therefore,
\begin{equation}
\delta^{\pi}_{ij}(\pi, \phi) \leq \frac{1}{n} \sum_{k,l = 1}^n \sum_{u = 1}^m q_{u \, \pi(u) \, k \, l} + \frac{1}{n} \sum_{k,l = 1}^n \left(q_{i \, \pi(j) \, k \, l} + q_{j \, \pi(i) \, k \, l} - q_{i \, \pi(i) \, k \, l} - q_{j \, \pi(j) \, k \, l}\right) - f(\pi, \phi). \nonumber
\end{equation}
Analogous result can be derived for similarly defined $\delta^{\phi}_{kl}$:
\begin{equation}
\delta^{\phi}_{kl}(\pi, \phi) \leq \frac{1}{m} \sum_{i,j = 1}^m \sum_{v = 1}^n q_{i \, j \, v \, \phi(v)} + \frac{1}{m} \sum_{i,j = 1}^m \left(q_{i \, j \, k \, \phi(l)} + q_{i \, j \, l \, \phi(k)} - q_{i \, j \, k \, \phi(k)} - q_{i \, j \, l \, \phi(l)}\right) - f(\pi, \phi). \nonumber
\end{equation}

We can now get an upper bound on the average cost change after optimized $2$-exchange operation on solution $(\pi, \phi)$.
\begin{align}
\Delta(\pi, \phi) \, &= \frac{\sum_{i,j=1}^m \delta^{\pi}_{ij}(\pi, \phi) + \sum_{k,l=1}^n \delta^{\phi}_{kl}(\pi, \phi)}{m^2 + n^2} \nonumber\\
& \leq \frac{\frac{m^2}{n} \sum_{u = 1}^m \sum_{k,l = 1}^n q_{u \, \pi(u) \, k \, l} + \frac{2}{n} \sum_{i,j = 1}^m \sum_{k,l = 1}^n q_{i \, j \, k \, l} - \frac{2m}{n} \sum_{i = 1}^m \sum_{k,l = 1}^n q_{i \, \pi(i) \, k \, l} - m^2 f(\pi, \phi)}{m^2 + n^2} \nonumber\\
& \quad + \frac{\frac{n^2}{m} \sum_{i,j = 1}^m \sum_{v = 1}^n q_{i \, j \, v \, \phi(v)} + \frac{2}{m} \sum_{i,j = 1}^m \sum_{k,l = 1}^n q_{i \, j \, k \, l} - \frac{2n}{m} \sum_{i,j = 1}^m \sum_{k = 1}^n q_{i \, j \, k \, \phi(k)} - n^2 f(\pi, \phi)}{m^2 + n^2} \nonumber\\
& = \frac{(m^3 - 2m^2) \sum_{i = 1}^m \sum_{k,l = 1}^n q_{i \, \pi(i) \, k \, l} + (n^3 - 2n^2) \sum_{i,j = 1}^m \sum_{v = 1}^n q_{i \, j \, v \, \phi(v)}}{mn(m^2 + n^2)} \nonumber\\
& \quad + \frac{2(m + n) \sum_{i,j = 1}^m \sum_{k,l = 1}^n q_{i \, j \, k \, l}}{mn(m^2 + n^2)} - f(\pi, \phi) \nonumber\\
& \leq \mu - f(\pi, \phi), \nonumber
\end{align}
where
\begin{equation}
\mu = \max_{\pi \in \Pi, \phi \in \Phi} \left[\dfrac{m^3 \sum_{i = 1}^m \sum_{k,l = 1}^n q_{i \, \pi(i) \, k \, l} + n^3 \sum_{i,j = 1}^m \sum_{v = 1}^n q_{i \, j \, v \, \phi(v)} + 2(m + n) \sum_{i,j = 1}^m \sum_{k,l = 1}^n q_{ijkl}}{mn(m^2 + n^2)}\right]. \nonumber
\end{equation}
Note that $\mu$ does not depend on any particular solution and is fixed for a given BAP instance.

For any given solution $(\pi, \phi)$ to BAP, either $f(\pi, \phi) \leq \mu$ or $f(\pi, \phi) > \mu$, which means that $\Delta(\pi, \phi) \leq 0$, and so there exists an optimized $2$-exchange operation that improves our solution cost by at least $f(\pi, \phi) - \mu$, thus, making it not worse than $\mu$. We also notice that,
\begin{align}
\mu - (m + n) \Aa &= \mu - \frac{m + n}{mn} \sum_{i,j=1}^m \sum_{k,l=1}^n q_{ijkl} = \mu - \frac{(m + n)(m^2 + n^2)}{mn(m^2 + n^2)} \sum_{i,j=1}^m \sum_{k,l=1}^n q_{ijkl} \nonumber\\
& = \max_{\pi \in \Pi} \left[\dfrac{m^3 \sum_{i = 1}^m \sum_{k,l = 1}^n q_{i \, \pi(i) \, k \, l}}{mn(m^2 + n^2)}\right] - \frac{m^3 \sum_{i,j=1}^m \sum_{k,l=1}^n q_{ijkl}}{mn(m^2 + n^2)} \nonumber\\
& + \max_{\phi \in \Phi} \left[\dfrac{n^3 \sum_{i,j = 1}^m \sum_{v = 1}^n q_{i \, j \, v \, \phi(v)}}{mn(m^2 + n^2)}\right] - \frac{n^3 \sum_{i,j=1}^m \sum_{k,l=1}^n q_{ijkl}}{mn(m^2 + n^2)} \nonumber\\
& + \dfrac{2(m + n) \sum_{i,j = 1}^m \sum_{k,l = 1}^n q_{ijkl}}{mn(m^2 + n^2)} - \frac{(m^2n + n^2m) \sum_{i,j=1}^m \sum_{k,l=1}^n q_{ijkl}}{mn(m^2 + n^2)} \leq 0, \nonumber\\
\end{align}
and so $(m + n) \Aa \geq \mu$, which completes the proof.
\end{proof}

We now show that by exploiting the properties of optimized $h$-exchange neighborhood, one can obtain a solution with an improved domination number, compared to the result in Theorem \ref{thm:dom}.

\begin{theorem}
For an integer $h$, a feasible solution to BAP, which is no worse than $\Omega((m - 1)!(n - 1)! + m^hn! + n^hm!)$ feasible solutions, can be found in $O(m^hn^3 + n^hm^3)$ time.
\end{theorem}
\begin{proof}
We show that the solution described in the statement of the theorem, can be obtained in the desired running time by choosing the best solution in the optimized $h$-exchange neighborhood of a solution with objective function value no worse than $\Aa(Q,C,D)$.

Let $(\mb{x}^*,\mb{y}^*)\in \ff$ be a BAP solution such that $f(\mb{x}^*,\mb{y}^*)\leq \Aa(Q,C,D)$.
Solution like that can be found in $O(m^2n^2)$ time using Theorem \ref{thm:minmax}. From the proof of Theorem \ref{thm:dom} we know that there exists a set $R_\sim$ of $(m - 1)!(n - 1)!$ solutions, with one solution from every class defined by the equivalence relation $\sim$, such that $f(\mb{x},\mb{y})\geq\Aa(Q,C,D)\geq f(\mb{x}^*,\mb{y}^*)$ for every $(\mb{x},\mb{y})\in R_\sim$.
Let $R_x$ denote the $[h,n]$-exchange neighborhood of $(\mb{x}^*,\mb{y}^*)$, and let $R_y$ denote the $[m,h]$-exchange neighborhood of $(\mb{x}^*,\mb{y}^*)$.
Note that $R_x\cup R_y$ is the optimized $h$-exchange neighborhood of $(\mb{x}^*,\mb{y}^*)$. $R_x\cup R_y$ can be searched in $O(m^hn^3 + n^hm^3)$ time, and the result of the search has the objective function value less or equal than every $(\mb{x},\mb{y})\in R_\sim \cup R_x \cup R_y$.

Consider $R'_x \subset R_x$  ($R'_y \subset R_y$) to be the set of solutions constructed in the same way as $R_x$  ($R_y$), but now only considering those reassignments of $h$-sets $S \in M$ ($S \in N$) that are different from $\mb{x}^*$ ($\mb{y}^*$) on entire $S$. By simple enumerations it can be shown that $|R'_x|=\binom{m}{h}(!h)n!$, $|R'_y|=\binom{n}{h}(!h)m!$ and $|R'_x \cap R'_y|=\binom{m}{h}(!h)\binom{n}{h}(!h)$, where $!h$ denotes the number of derangements (i.e.\@ permutations without fixed points) of $h$ elements. Furthermore, $|R_\sim \cap R'_x|\leq \binom{m}{h}(!h)(n - 1)!$ and $|R_\sim \cap R'_y|\leq \binom{n}{h}(!h)(m - 1)!$. The later two inequalities are due to the fact that for some fixed $\mb{x}'$ ($\mb{y}'$), the relation $\sim$ partitions the set of solutions $\{\mb{x}'\}\times \yy$ ($\xx\times \{\mb{y}'\}$) into equivalence classes of size $n$ ($m$) exactly, and each such class contains at most one element of $R_\sim$.  Now we get that
\begin{align*}
|R_\sim \cup R_x\cup R_y| &\geq |R_\sim \cup R_x^d\cup R_y^d|\\
&\geq |R_\sim|+|R_x^d|+|R_y^d|-|R_\sim \cap R_x^d|-|R_\sim \cap R_y^d|-|R_x^d \cap R_y^d|\\
&\geq (m - 1)!(n - 1)! + \binom{m}{h}(!h)n! + \binom{n}{h}(!h)m!\\
&\ \ \ \ -\binom{m}{h}(!h)(n - 1)! - \binom{n}{h}(!h)(m - 1)! - \binom{m}{h}(!h)\binom{n}{h}(!h)\\
&\in \Omega((m - 1)!(n - 1)! + m^hn! + n^hm!),
\end{align*}
which concludes the proof.
\end{proof}

%%%%%%%%%%%%%%%%%%%%%%%%%%%%%%%%%%%%%%%%%%%%%%%%%%%%%%%%%%%%%%%%
%%%%%%%%%%%%%%%%%%%%%%%%%%%%%%%%%%%%%%%%%%%%%%%%%%%%%%%%%%%%%%%%
\subsection{Shift based neighborhoods}

Following the equivalence class example in Section \ref{sec:notations}, the \textit{shift} neighborhood of a given solution $(\mb{x}, \mb{y})$ will be comprised of all $m$ solutions $(\mb{x}', \mb{y})$, such that $x_{ij}'=x_{i(j + a \mod m)}, \forall a \in M$ and all $n$ solutions $(\mb{x}, \mb{y}')$, such that $y_{kl}'=y_{k(l + b \mod m)}, \forall b \in N$.
Alternatively, shift neighborhood can be described in terms of the permutation formulation of BAP. Given a permutation pair ($\pi$, $\phi$), we are looking at all $m$ solutions ($\pi'$, $\phi$), such that $\pi'(i) = \pi(i) + a \mod m, \forall a \in M$, and all $n$ solutions ($\pi$, $\phi'$), such that $\phi'(k) = \phi(k) + b \mod m, \forall b \in N$. Intuitively this means that, either $\pi$ will be cyclically shifted by $a$ or $\phi$ will be cyclically shifted by $b$, hence the name of this neighborhood. An iteration of the local search algorithm based on Shift neighborhood will take $O(mn^2)$ time, as we are required to fully recompute each of the $m$ (resp. $n$) solutions objective values.

Using the same asymptotic running time per iteration, it is possible to explore the neighborhood of a larger size, with the help of additional data structures $e_{ij}, g_{kl}$ (see Section \ref{sec:hex}) that maintain partial sums of assigning $i \in M$ to $j \in M'$ and $k \in N$ to $l \in N'$ given $\mb{y}$ and $\mb{x}$ respectively. Consider $\Theta(n^2)$ size neighborhood \textit{shift+shuffle} defined as follows. For a given permutation solution ($\pi$, $\phi$) this neighborhood will contain all ($\pi'$, $\phi$) such that
\begin{equation}
\pi'(i) = \pi\left((i \mod \lfloor \frac{m}{u} \rfloor) u + \lfloor \frac{i}{\lfloor \frac{m}{u} \rfloor} \rfloor + a \mod m\right), \quad \forall a \in M, \, \forall u \in \{1, 2, \ldots, \lfloor \frac{m}{2} \rfloor \},
\end{equation}
and all ($\pi$, $\phi'$) such that
\begin{equation}
\phi'(k) = \phi\left((k \mod \lfloor \frac{n}{v} \rfloor) v + \lfloor \frac{k}{\lfloor \frac{n}{v} \rfloor} \rfloor + b \mod n\right), \quad \forall b \in N, \, \forall v \in \{1, 2, \ldots, \lfloor \frac{n}{2} \rfloor \}.
\end{equation}
Two of the above equations are sufficient for the case of $m \mod u = 0$ or $n \mod v = 0$. Otherwise, for all $i > m - (m \mod u)$ and all $k > n - (n \mod v)$ an arbitrary reassignment could be applied (for example $\pi'(i)=\pi(i)$ and $\phi'(k)=\phi(k)$). One can visualize shuffle operation as splitting elements of a permutation into buckets of the same size ($u$ or $v$ in the formulas above), and then forming a new permutation by placing first elements from each bucket in the beginning, followed by second elements of each bucket, and so on. Figure \ref{shuffle} depicts such shuffling for a permutation $\pi$.
\begin{figure}[h!]
\centering
\includegraphics[width=0.6\textwidth]{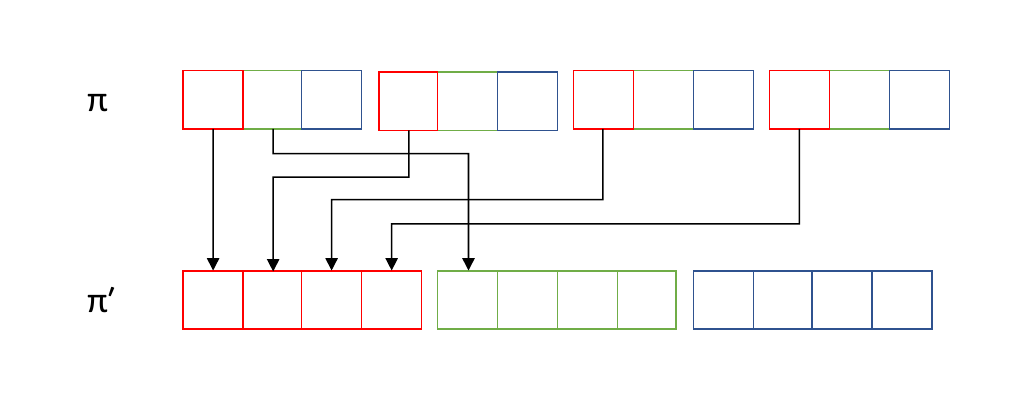}
\caption{Example of shuffle operation on permutation $\pi$, with $u=3$}
\label{shuffle}
\end{figure}
By combining shift and shuffle we increase the size of the explored neighborhood, at no extra asymptotic running time cost for the local search implementations.

Local search algorithms that explore shift or shift+shuffle neighborhoods could potentially be stuck in the arbitrarily bad local optimum, following the same argument as in Theorem \ref{thm:hex}.

\medskip

If we allow applying shift simultaneously to both $\mb{x}$ and $\mb{y}$ we will consider all $mn$ neighbors of the current solution, precisely as in equivalence class example from Section \ref{sec:notations}. We will call this \textit{dual shift} neighborhood of a solution $(\mb{x}, \mb{y})$. Notice that a local search algorithm that explores this neighborhood reaches a local optimum only after a single iteration, with running time $O(m^2n^2)$.

A much larger \textit{optimized shift} neighborhood will be defined as follows. For every shift operation on $\mb{x}$ we consider all possible assignments of $\mb{y}$, and vice versa, for each shift on $\mb{y}$ we will consider all possible assignments of $\mb{x}$. Just like in the case of optimized $h$-exchange, this neighborhood is exponential in size, but can be efficiently explored in $O(mn^3)$ running time by solving corresponding linear assignment problems.

\begin{theorem}
For local search based on dual shift and optimized shift neighborhoods, the final solution value is guaranteed to be no worse than $\Aa(Q,C,D)$.
\end{theorem}
\begin{proof}
The proof for dual shift neighborhood follows from the fact that we are completely exploring the equivalence class defined by $\sim$ of a given solution, as in Theorem \ref{thm:minmax}.

For optimized shift, notice that for each shift on one side of $(\mb{x}, \mb{y})$ we consider all possible solutions on the other side. This includes all possible shifts on that respective side. Therefore the set of solutions of optimized shift neigborhood includes the set of solutions of dual shift neighborhood, and contains the solution with the value at most $\Aa(Q,C,D)$.
\end{proof}

\bigskip

In \cite{custicbilinear} we have explored the complexity of a special case of BAP where $Q$, observed as a $m^2 \times n^2$ matrix, is restricted to be of a fixed rank. The rank of such $Q$ is said to be at most $r$ if and only if there exist some $m\times m$ matrices $A^{^{p}}=(a_{ij}^p)$ and $n\times n$ matrices  $B^{^p}=(b_{ij}^p)$, $p=1,\ldots,r,$ such that
\begin{equation}\label{fact}
q_{ijkl}=\sum_{p=1}^r a_{ij}^pb_{kl}^p
\end{equation}
for all $i,j\in M$, $k,l\in N$.

\begin{theorem}
Alternating Algorithm and local search algorithms that explore optimized $h$-exchange and optimized shift neighborhoods will find an optimal solution to BAP $(Q,C,D)$, if $Q$ is a non-negative matrix of rank $1$, and both $C$ and $D$ are zero matrices.
\end{theorem}
\begin{proof}
Note that in the case described in the statement of the theorem, we are looking for such $(\mb{x}^*, \mb{y}^*)$ that minimizes $(\sum_{i,j=1}^m a_{ij}x^*_{ij}) \cdot (\sum_{k,l=1}^n b_{kl} y^*_{kl})$, where $q_{ijkl}=a_{ij} b_{kl}, \, \forall i,j\in M, \, k,l\in N$. If we are restricted to non-negative numbers, solutions to corresponding linear assignment problems would be an optimal solution to this BAP. It is easy to see that, for any fixed $\mb{x}$, a solution of the smallest value will be produced by $\mb{y}^*$. And viceversa, for any fixed $\mb{y}$, a solution of the smallest value will be produced by $\mb{x}^*$.

Optimized $h$-exchange neighborhood, optimized shift neighborhood and the neighborhood that \textit{Alternating Algorithm} is based on, all contain the solution that has one side of $(\mb{x}, \mb{y})$ unchanged and has the optimal assignment on the other side. Therefore, the local search algorithms that explore these neighborhoods will proceed to find optimal $(\mb{x}^*, \mb{y}^*)$ in at most $2$ iterations.
\end{proof}

%%%%%%%%%%%%%%%%%%%%%%%%%%%%%%%%%%%%%%%%%%%%%%%%%%%%%%%%%%%%%%%%
%%%%%%%%%%%%%%%%%%%%%%%%%%%%%%%%%%%%%%%%%%%%%%%%%%%%%%%%%%%%%%%%
%%%%%%%%%%%%%%%%%%%%%%%%%%%%%%%%%%%%%%%%%%%%%%%%%%%%%%%%%%%%%%%%
%%%%%%%%%%%%%%%%%%%%%%%%%%%%%%%%%%%%%%%%%%%%%%%%%%%%%%%%%%%%%%%%
\section{Experimental design and test problems}
\label{sec:expsetup}

In this section we present general information on the design of our experiments and generation of test problems.

\medskip

All experiments are conducted on a PC with Intel Core i7-4790 processor, 32 GB of memory under control of Linux Mint 17.3 (Linux Kernel 3.19.0-32-generic) 64-bit operating system. Algorithms are coded using Python 2.7 programming language and run via PyPy 5.3 implementation of Python. The linear assignment problem, that appears as a subproblem for several algorithms, is solved using Hungarian algorithm \cite{kuhn1955hungarian} implementation in Python.

%%%%%%%%%%%%%%%%%%%%%%%%%%%%%%%%%%%%%%%%%%%%%%%%%%%%%%%%%%%%%%%%
%%%%%%%%%%%%%%%%%%%%%%%%%%%%%%%%%%%%%%%%%%%%%%%%%%%%%%%%%%%%%%%%
\subsection{Test problems}

As there are no existing benchmark instances available for BAP, we have created several sets of test problems, which could be used by other researchers in the future experimental analysis. Three categories of problem instances are considered: \textit{\textbf{uniform}}, \textit{\textbf{normal}} and \textit{\textbf{euclidean}}.

\begin{itemize}
\item For \textit{uniform} instances we set $c_{ij}, d_{kl} = 0$ and the values $q_{ijkl}$ are generated randomly with uniform distribution from the interval $[0, m n]$ and rounded to the nearest integer.
\item For \textit{normal} instances we set $c_{ij}, d_{kl} = 0$ and the values $q_{ijkl}$ are generated randomly following normal distribution with mean $\mu = \frac{m n}{2}$, standard deviation $\sigma = \frac{m n}{6}$ and rounded to the nearest integer.
\item For \textit{euclidean} instances we generate randomly with uniform distribution four sets of points $A,B,U,V$ in Euclidean plane of size $[0, 1.5 \sqrt[2]{m n}] \times [0, 1.5 \sqrt[2]{m n}]$, such that $|A| = |B| = m$, $|U| = |V| = n$. Then $C$ and $D$ are chosen as zero vectors, and $q_{ijkl} = ||a_i - u_k|| \cdot ||b_j - v_l||$ (rounded to the nearest integer), where $a_i \in A, b_j \in B, u_k \in U, v_l \in V$.
% microarray qap instances: Since QAP is a special case of BAP, we tested our algorithms on this class of problems. We selected the benchmark library of "microarray QAP" instances.
\end{itemize}

% \begin{question}
% \label{qeucl}
% Is \textit{Euclidean} BAP NP-hard?
% \end{question}

% Question \ref{qeucl} is an interesting issue, as some related results have been studied for QAP. Namely, for Koopmans-Beckmann QAP with one of the two matrices being a distance matrix that satisfy triangle inequality and the other matrix having only two distinct values, the problem remains NP-hard and inaproximable in polynomial time \cite{que1986}. However, there is no corresponding results for QAP that would explore complexity and approximalbility for the case when both matrices represent euclidean distances between points.

\medskip

Test problems are named using the convention ``type size number'', where type $\in$ \{\textit{uniform}, \textit{normal}, \textit{euclidean}\}, size is of the form $m \times n$, and number $\in \{0, 1, \ldots\}$.
For every instance type and size we have generated 10 problems, and all the results of experiments will be averaged over those 10 problems. For example, in a table or a figure, a data point for ``uniform $50 \times 50$'' would be the average among the 10 generated instances. This applies to objective function values, running times and number of iterations, and would not be explicitly mentioned throughout the rest of the paper. Problem instances, results for our final set of experiments as well as best found solutions for every instance are available upon request from Abraham Punnen (apunnen@sfu.ca).
% \url{https://vault.sfu.ca/?}.
%To obtain the precise value for each instance we refer to the supplementary page [] where additional information can be obtained. These include individual objective values for each instance and algorithms, test problem generator, best found solutions.

%%%%%%%%%%%%%%%%%%%%%%%%%%%%%%%%%%%%%%%%%%%%%%%%%%%%%%%%%%%%%%%%
%%%%%%%%%%%%%%%%%%%%%%%%%%%%%%%%%%%%%%%%%%%%%%%%%%%%%%%%%%%%%%%%
%%%%%%%%%%%%%%%%%%%%%%%%%%%%%%%%%%%%%%%%%%%%%%%%%%%%%%%%%%%%%%%%
%%%%%%%%%%%%%%%%%%%%%%%%%%%%%%%%%%%%%%%%%%%%%%%%%%%%%%%%%%%%%%%%
\section{Experimental analysis of construction heuristics}
\label{sec:expconstr}

In Section \ref{sec:constr} we presented several construction approaches to generate a solution to BAP. In this section we discuss results of computational experiments using these heuristics.

\medskip

The experimental results are summarized in Table \ref{constrt}.
For the heuristic \textit{GreedyRandomized}, we have considered the candidate list size $2$, $4$ and $6$. In the table, columns GreedyRandomized2 and GreedyRandomized4 refer to implementations with candidate list size of $2$ and $4$, respectively. Results for candidate list size $6$ are excluded from the table due to poor performance.

Here and later when presenting computational results, ``value'' and ``time'' refer to objective function value and running time of an algorithm. The best solution value among all tested heuristics is shown in bold font.
We also report (averaged over 10 instances of given type and size) the average solution value $\Aa(Q,C,D)$ (denoted simply as $\Aa$), computed using the closed-form expression from Section \ref{sec:notations}.

\input{tables//constr.tex}

As the table shows, for smaller \textit{uniform} and \textit{normal} instances as well as for all \textit{euclidean} instances \textit{Rounding} produced better quality results, however, using substantially longer time. For all other problems \textit{RandomXYGreedy} obtained better results.
To our surprise, the quality of the solution produced by \textit{Greedy} was inferior to that of \textit{RandomXYGreedy}. It can, perhaps, be explained as a consequence of being ``too greedy'' in the beginning, leading to worse overall solution, particularly, taking into consideration the quadratic nature of the objective function.
In the initial steps the choice is made based on the very much incomplete information about solution and the interaction cost of $\mb{x}$ and $\mb{y}$ assignments.
In addition, the running time for \textit{RandomXYGreedy} was significantly lower than that of \textit{Rounding} and other algorithms. Thus, we conclude that \textit{RandomXYGreedy} is our method of choice if a solution to BAP is needed quickly.

As for the \textit{GreedyRandomized} strategy, the higher the size of the candidate list, the worse is the quality of the resulting solution.
On the other hand, larger sizes of the candidate lists provide us with more diversified ways to generate solutions for BAP. That may have advantages if the construction is followed by an improvement approach as generally done in GRASP algorithm.

In Figures \ref{construv} and \ref{construt} we present solution value and running time results of this section for \textit{uniform} instances.

\begin{figure}[h!]
\centering
\includegraphics[width=0.99\textwidth]{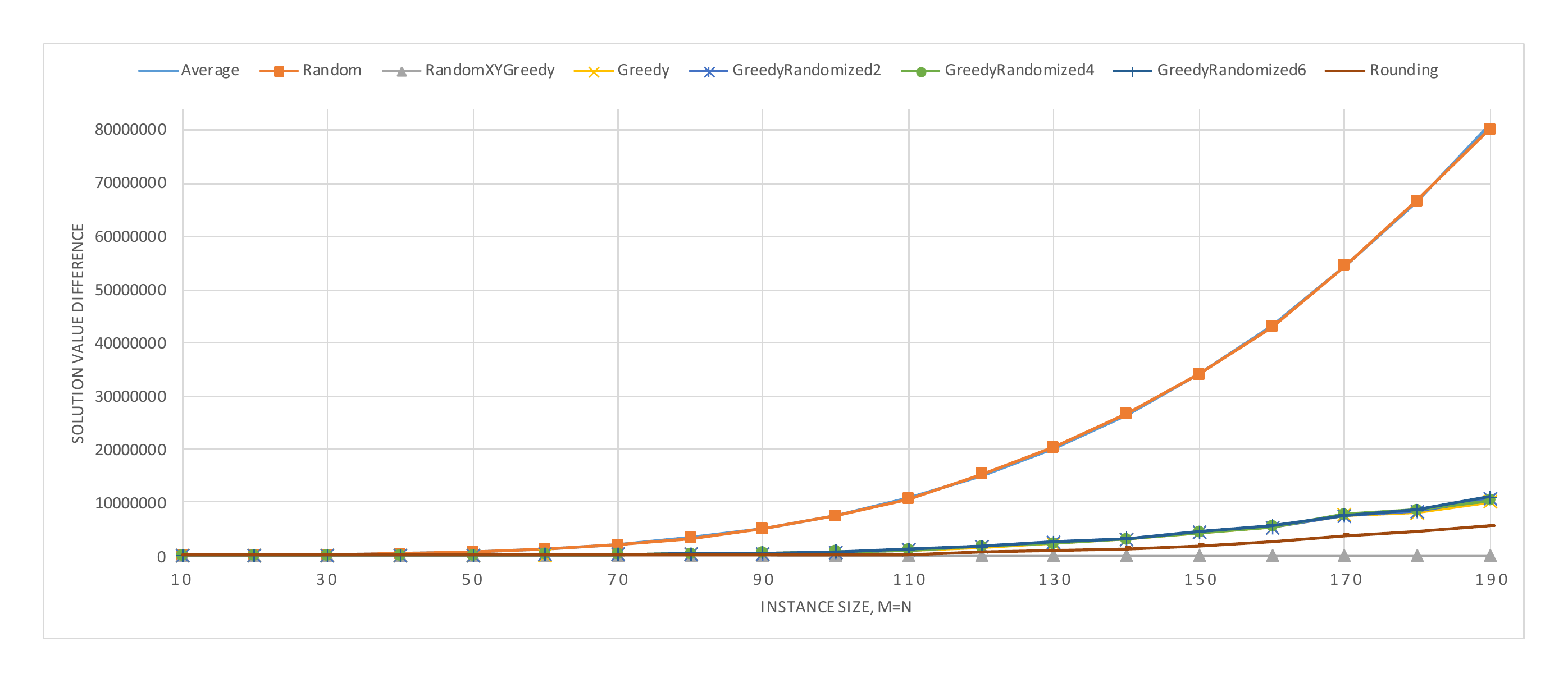}
\caption{Difference between solution values (to the best) for construction heuristics; \textit{uniform} instances}
\label{construv}
\end{figure}

\begin{figure}[h!]
\centering
\includegraphics[width=0.99\textwidth]{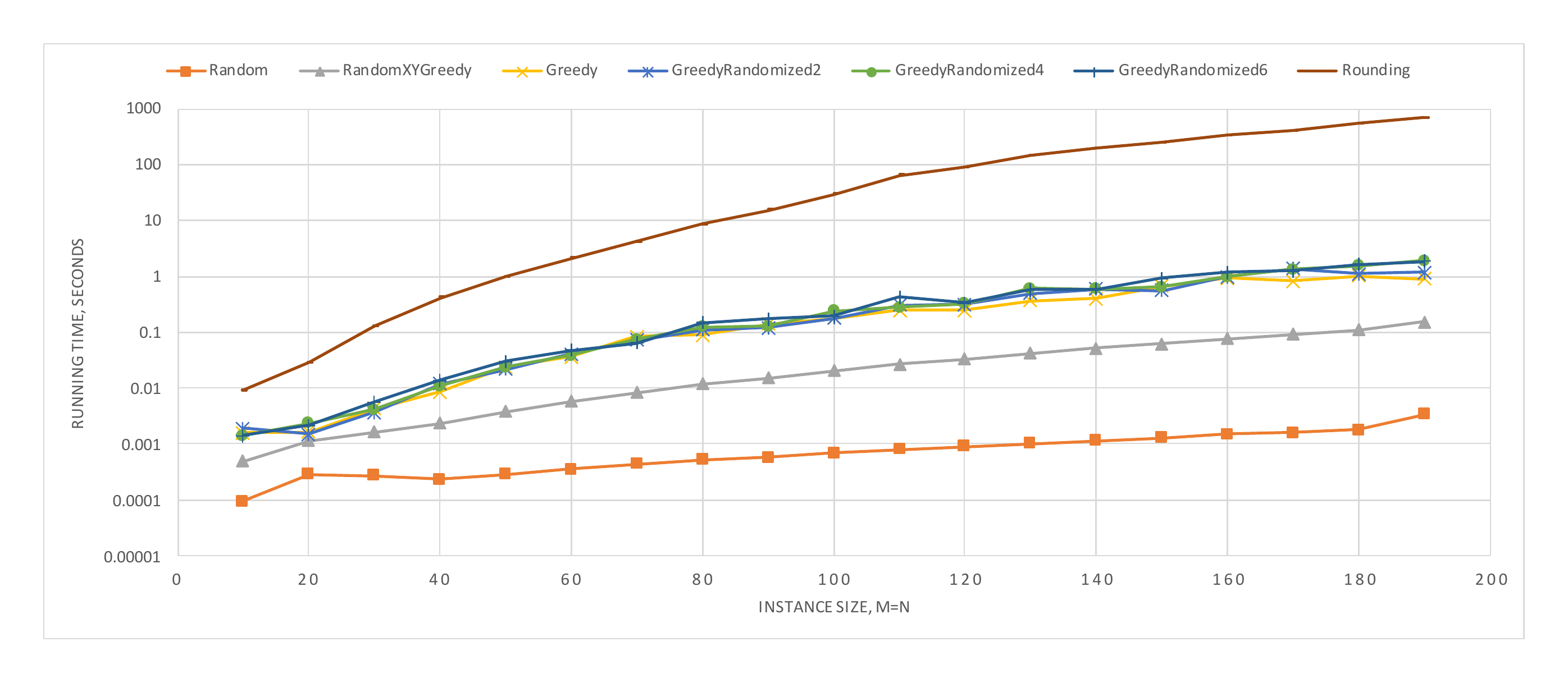}
\caption{Running time for construction heuristics; \textit{uniform} instances}
\label{construt}
\end{figure}

%%%%%%%%%%%%%%%%%%%%%%%%%%%%%%%%%%%%%%%%%%%%%%%%%%%%%%%%%%%%%%%%
%%%%%%%%%%%%%%%%%%%%%%%%%%%%%%%%%%%%%%%%%%%%%%%%%%%%%%%%%%%%%%%%
%%%%%%%%%%%%%%%%%%%%%%%%%%%%%%%%%%%%%%%%%%%%%%%%%%%%%%%%%%%%%%%%
%%%%%%%%%%%%%%%%%%%%%%%%%%%%%%%%%%%%%%%%%%%%%%%%%%%%%%%%%%%%%%%%
\section{Experimental analysis of local search algorithms}
\label{sec:expls}

Let us now discuss the results of computational experiments carried out using local search algorithms that explore neighborhoods discussed in Section \ref{sec:ls}. All algorithms are started from the same random solution and ran until a local optimum is reached. In addition to the objective function value and running time we report the number of iterations for each approach.

\medskip

For $h$-exchange neighborhoods, we selected $2$ and $3$-exchange local search algorithms (denoted by \textit{\textbf{2ex}} and \textit{\textbf{3ex}}) as well as the Alternating Algorithm (\textit{\textbf{AA}}).

From [$h,p$]-exchange based algorithms, we have implemented $[2,2]$-exchange local search (named \textit{\textbf{Dual2ex}}). The $[2,2]$-exchange neighborhood can be explored in $O(m^2n^2)$ time, using efficient recomputation of the change in the objective value. We refer to the algorithm that explores optimized $2$-exchange neighborhood as \textit{\textbf{2exOpt}}. The running time of each iteration of this local search is $O(m^2 n^3)$. To speed up this potentially slow approach, we have also considered a version, namely \textit{\textbf{2exOptHeuristic}}, where we use an $O(n^2)$ heuristic to solve the underlying linear assignment problem, instead of the Hungarian algorithm with cubic running time. The running time of each iteration of 2exOptHeuristic is then $O(m^2 n^2)$. Similarly defined will be \textit{\textbf{3exOpt}}.

\textit{\textbf{Shift}}, \textit{\textbf{ShiftShuffle}}, \textit{\textbf{DualShift}} and \textit{\textbf{ShiftOpt}} are implementations of local search based on shift, shift+shuffle, dual shift and optimized shift neighborhoods respectively.

In addition, we consider variations of the above-mentioned algorithms, namely \textit{\textbf{2exFirst}}, \textit{\textbf{3exFirst}}, \textit{\textbf{Dual2exFirst}}, \textit{\textbf{2exOptFirst}}, \textit{\textbf{2exOptHeuristicFirst}}, \textit{\textbf{ShiftOptFirst}}, where corresponding neighborhoods explored only until the first improving solution is encountered.

\medskip

We provide a summary of complexity results on these local search algorithms in Table \ref{sum}. Here by $I$ we denote the number of iterations (or ``moves'') that it takes for a corresponding search to converge to a local optimum. As $I$ could potentially be exponential in $n$ and will vary between algorithms, we use this notation to simply emphasize the running time of an iteration of each approach.

\input{tables//summary.tex}

\medskip

Table \ref{convt} summarizes experimental results for \textit{2ex}, \textit{3ex}, \textit{AA}, \textit{2exOpt} and \textit{2exOptFirst}. Results for other algorithms are not included in the table due to inferior performance. However, figures \ref{convuv} and \ref{convut} provide additional insight into the performance of all the algorithms we have tested, for the case of \textit{uniform} instances.

\input{tables//conv.tex}

\begin{figure}[h!]
\centering
\includegraphics[width=0.99\textwidth]{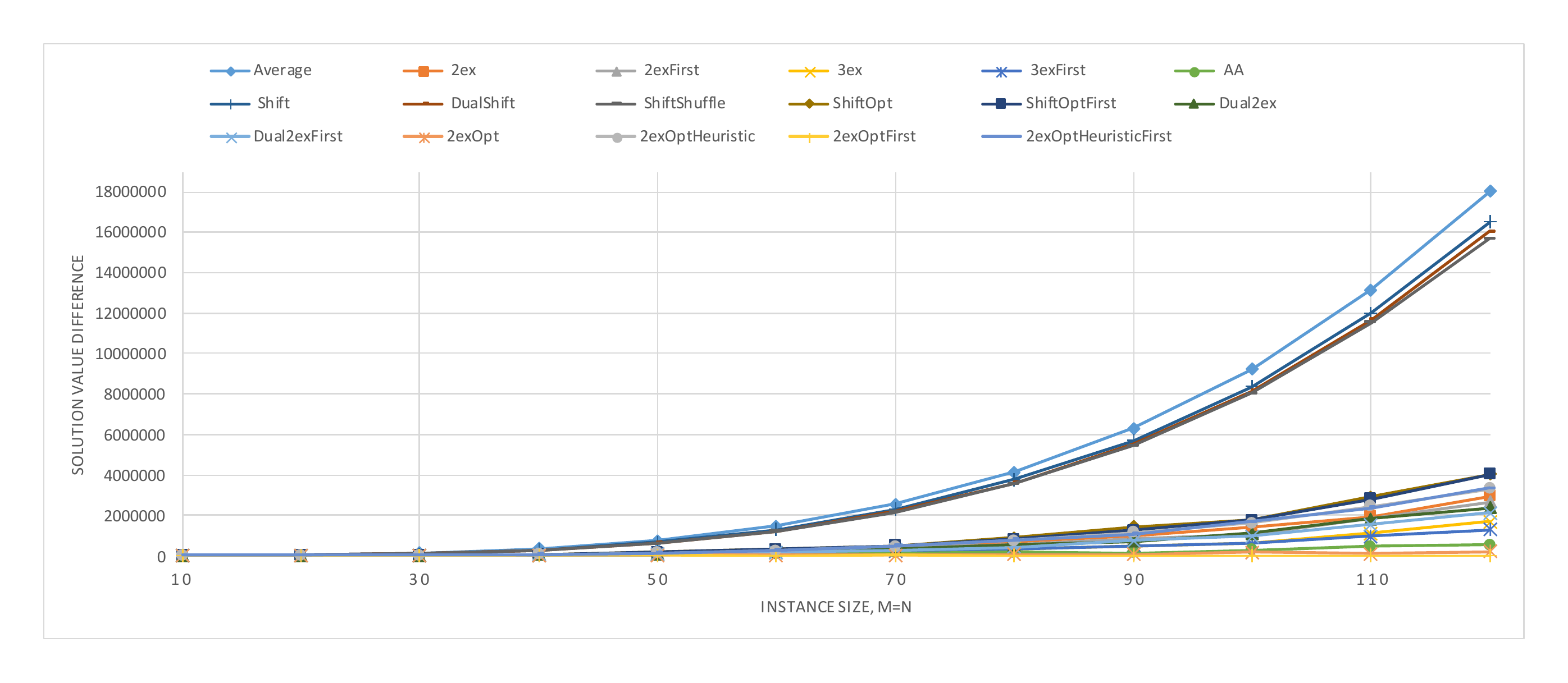}
\caption{Difference between solution values (to the best) for local search; \textit{uniform} instances}
\label{convuv}
\end{figure}

\begin{figure}[h!]
\centering
\includegraphics[width=0.99\textwidth]{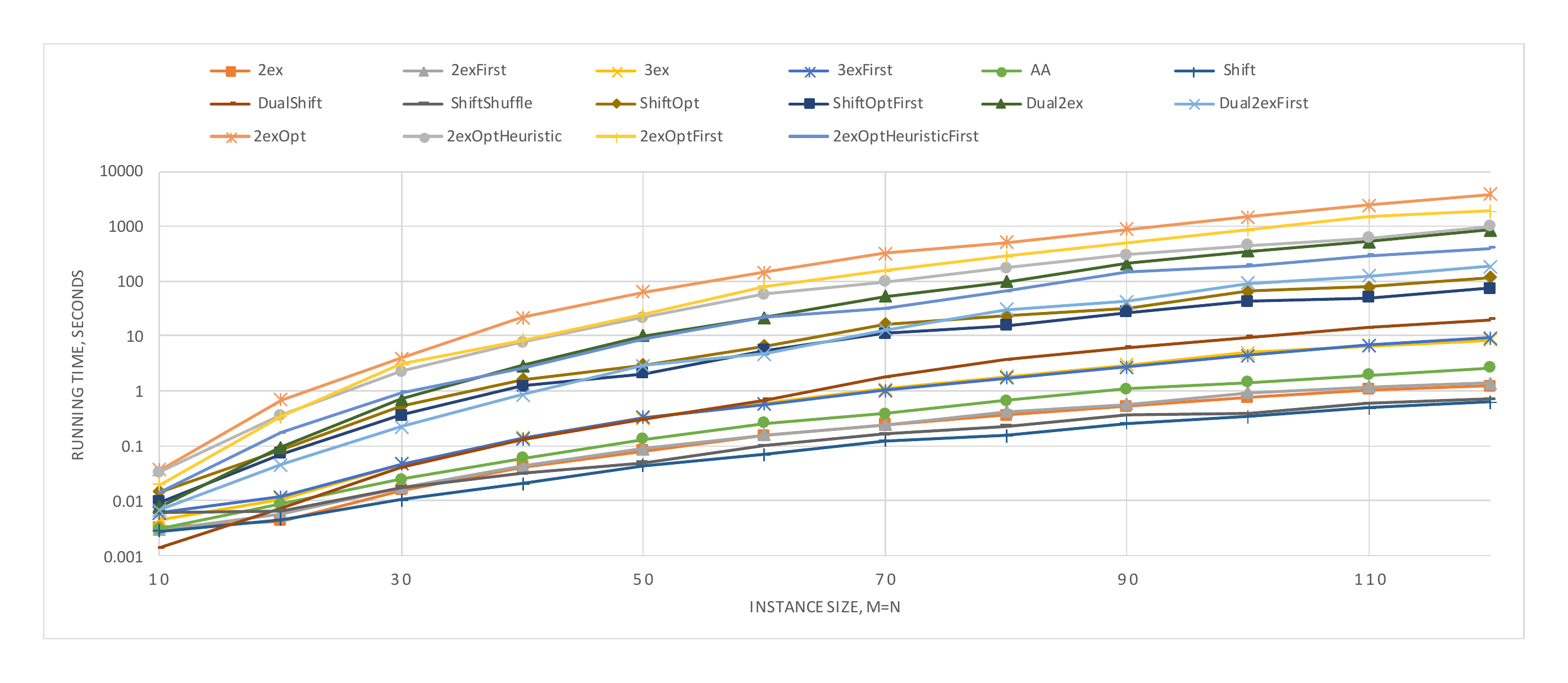}
\caption{Running time to converge for local search; \textit{uniform} instances}
\label{convut}
\end{figure}

Even though the convergence speed is very fast for implementations of \textit{Shift}, \textit{ShiftShuffle} and \textit{DualShift}, the resulting solution values are not significantly better than the average value $\Aa(Q,C,D)$ for the instance.

The \textit{optimized shift} versions, namely \textit{ShiftOpt} and \textit{ShiftOptFirst} produced better solutions but still are outperformed by all remaining heuristics. This fact together with the slower convergence speed (as compared to say \textit{2ex}) shows the weaknesses of the approach.

\textit{Dual2ex} and \textit{Dual2exFirst} are heavily outperformed both in terms of convergence speed as well as the quality of the resulting solution by \textit{AA}.

It is also worth mentioning that speeding up \textit{2exOpt} and \textit{2exOptFirst} by substituting the Hungarian algorithm with an $O(n^2)$ heuristic for the assignment problem did not provide us with good results. The solution quality decreased substantially and, considering that the running time to converge is still slower than that of \textit{AA}, we discard these options.

Table \ref{convt} presents the results for the better performing set of algorithms. The performance of both \textit{first improvement} and \textit{best improvement} approaches \textit{2exFirst}, \textit{3exFirst} and \textit{2ex}, \textit{3ex} respectively are similar so we will consider only the latter two from now on. Interestingly, it is not the case for the \textit{optimized} neighborhoods. We noticed that, for \textit{uniform} and \textit{normal} instances \textit{2exOptFirst} runs faster than \textit{2exOpt}, in most cases. However, for \textit{euclidean} instances \textit{2exOptFirst} takes more time to converge.

As expected, \textit{AA} is better than \textit{3ex} with respect to both solution quality and running time. We will not include any of the $h$-exchange neighborhood search implementations for $h > 3$ in this study due to relatively poor performance and huge running time.

We focused the remaining experiments in the paper on \textit{2ex}, \textit{AA} and \textit{2exOpt}. Among these \textit{2ex} converges the fastest, \textit{2exOpt} provides the best solutions and \textit{AA} assumes a ``balanced'' position. It is also clear that even better solution quality could be achieved by using implementations of optimized $h$-exchange neighborhood search with higher $h$. However, we show in the next sub-section that this is not feasible in terms of efficient metaheuristics implementation.

%%%%%%%%%%%%%%%%%%%%%%%%%%%%%%%%%%%%%%%%%%%%%%%%%%%%%%%%%%%%%%%%
%%%%%%%%%%%%%%%%%%%%%%%%%%%%%%%%%%%%%%%%%%%%%%%%%%%%%%%%%%%%%%%%
\subsection{Local search with multi-start}
\label{sec:explsms}

Now we would like to see how well our heuristics perform in terms of solutions quality, when the amount of time is fixed. For this we implemented a simple multi-start strategy for each of the algorithms. The framework will keep restarting the local search from the new \textit{Random} instance until the time limit is reached. The best solution found in the process is then reported as the result.

Time limit for each instance will be set as the following. Considering the results of the previous sub-section, we expect \textit{3exOptFirst} to be the slowest method to converge for all of the instances. We run it exactly once, and use its running time as a time limit for other multi-start algorithms. Together with resulting values we also report the number of restarts of each approach in Table \ref{mst}. Clearly, the choice of time limit yields $1$ as the number of starts for \textit{3exOptFirst}.

\input{tables//ms.tex}

\begin{figure}[h!]
\centering
\includegraphics[width=0.99\textwidth]{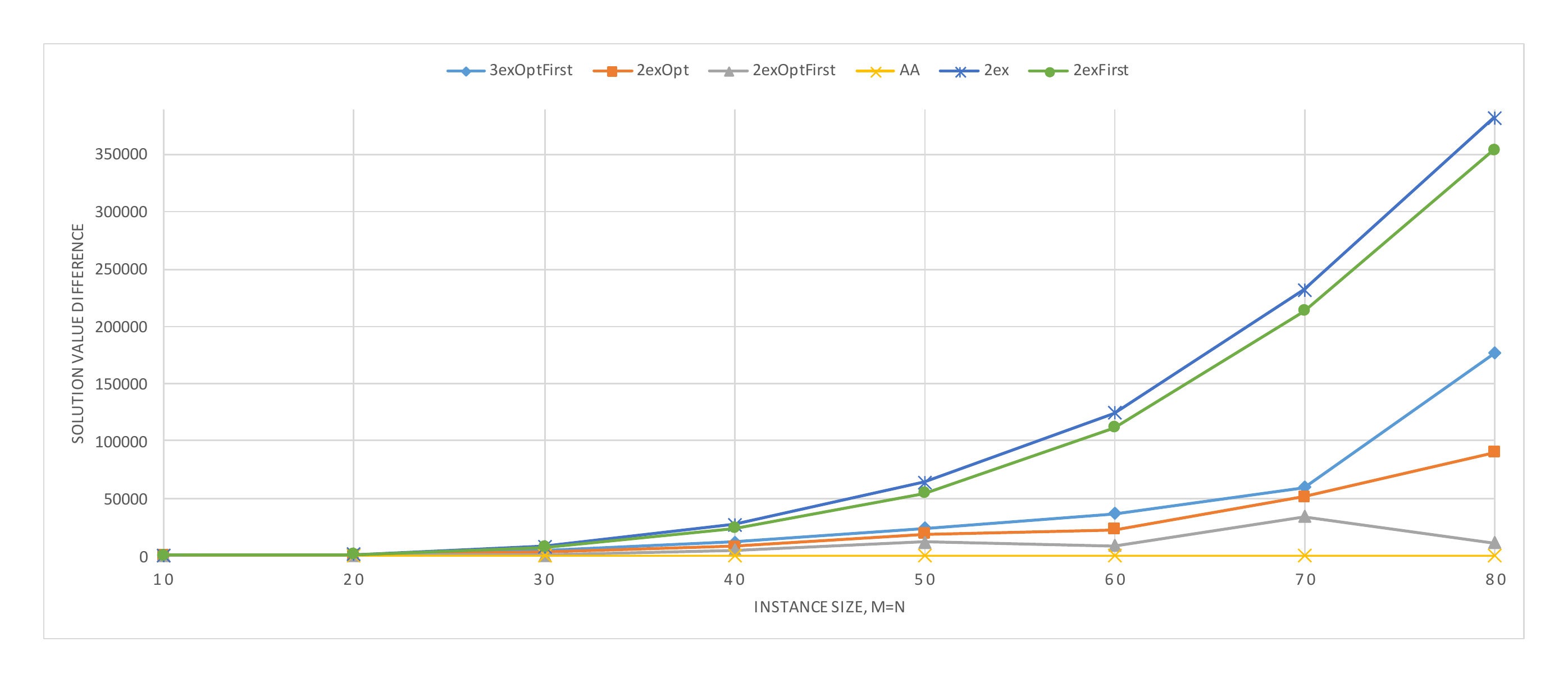}
\caption{Difference between solution values (to the best) for multi-start algorithms; \textit{uniform} instances}
\label{msv}
\end{figure}

The best algorithm in these settings is \textit{AA}, which consistently exhibited better performance for all instance types. The reason behind this is the fact that a local optimum by this approach can be reached almost as fast as by \textit{2ex}, however solution quality is much better. On the other hand, the convergence of \textit{2exOpt} to a local optimum is very time consuming, and perhaps a better strategy is to do more restarts with slightly less quality of resulting solution. Similar argument holds for the case why \textit{2exOptFirst} outperforms \textit{3exOptFirst} in this type of experiments. This observation is in contrast with the results experienced by researches of bipartite unconstrained binary quadratic program \cite{glover2015integrating} and bipartite quadratic assignment problem \cite{punnen2016bipartite}. The difference can be attributed to the more complex structure of BAP in comparison to problems mentioned above.

%%%%%%%%%%%%%%%%%%%%%%%%%%%%%%%%%%%%%%%%%%%%%%%%%%%%%%%%%%%%%%%%
%%%%%%%%%%%%%%%%%%%%%%%%%%%%%%%%%%%%%%%%%%%%%%%%%%%%%%%%%%%%%%%%
%%%%%%%%%%%%%%%%%%%%%%%%%%%%%%%%%%%%%%%%%%%%%%%%%%%%%%%%%%%%%%%%
%%%%%%%%%%%%%%%%%%%%%%%%%%%%%%%%%%%%%%%%%%%%%%%%%%%%%%%%%%%%%%%%
\section{Variable neighborhood search}
\label{sec:expvnsms}

Variable neighborhood search (VNS) is an algorithmic paradigm to enhance standard local search by making use of properties (often complementary) of multiple neighborhoods \cite{ahuja2007very,hansen2016variable}. The $2$-exchange neighborhood is very fast to explore and optimized $2$-exchange is more powerful but searching through it for an improving solution takes significantly more time. The neighborhood considered in the \textit{Alternating Algorithm} works better when significant asymmetry is present regarding $\mb{x}$ and $\mb{y}$ variables. Motivated by these complementary properties, we have explored VNS based algorithms  to solve BAP.

We start by attempting to improve the convergence speed of \textit{AA} by the means of the faster \textit{2ex}. The first variation, named \textbf{\textit{2ex+AA}} will first apply \textit{2ex} to \textit{Random} starting solution and then apply \textit{AA} to the resulting solution. A more complex approach \textbf{\textit{2exAAStep}} (Algorithm \ref{2exAAStep}) will start by applying \textit{2ex} and as soon as the search converge it will apply a single improvement (step) with respect to \textit{Alternating Algorithm} neighborhood. After successful update the procedure defaults to running \textit{2ex} again. The process stops when no more improvements by \textit{AA} (and consequently by \textit{2ex}) are possible.

\begin{algorithm}
\caption{$2exAAStep$}
\label{2exAAStep}
\begin{algorithmic}[0]\scriptsize
\Input integers $m, n$; $m \times m \times n \times n$ array $Q$; feasible solution $(\mb{x}, \mb{y})$ to given BAP
\Output feasible solution to given BAP

\While{True}

\State $(\mb{x}, \mb{y}) \gets 2ex(m, n, Q, (\mb{x}, \mb{y}))$ \Comment{running $2$-exchange local search (Section \ref{sec:hex})}

\State $e_{ij} \gets \sum_{k, l \in N} q_{ijkl} y_{kl} \, \forall i, j \in M$
\State $\mb{x}^* \gets arg\,min_{\mb{x}' \in \xx} \sum_{i, j \in M} e_{ij} x'_{ij}$ \Comment{solving assignment problem for $\mb{x}$}

\If{$f(\mb{x}^*, \mb{y}) < f(\mb{x}, \mb{y})$ }
\State \textbf{continue} \Comment{restarting the procedure \textbf{while} loop}
\EndIf

\State $g_{kl} \gets \sum_{i, j \in M} q_{ijkl} x^*_{ij} \, \forall k, l \in N$
\State $\mb{y}^* \gets arg\,min_{\mb{y}' \in \yy} \sum_{k, l \in N} g_{kl} y'_{kl}$ \Comment{solving assignment problem for $\mb{y}$}

\If{$f(\mb{x}^*, \mb{y}^*) = f(\mb{x}, \mb{y})$ }
\State \textbf{break} \Comment{algorithm converged, terminate}
\EndIf

\State $\mb{x} \gets \mb{x}^*; \, \mb{y} \gets \mb{y}^*$

\EndWhile
\State \textbf{return} ($\mb{x}$, $\mb{y}$)
\end{algorithmic}
\end{algorithm}

Results in Table \ref{vnsconvt1} follow the structure of experimental results reported earlier in the paper. The number of iterations that we report for \textit{2exAAStep} is the number of times the heuristic switches from $2$-exchange neighborhood to the neighborhood of the \textit{Alternating Algorithm}. Clearly, this number will be $1$ for \textit{2ex+AA} by design.

\input{tables//vnsconv1.tex}

As all these approaches are guaranteed to be locally optimal with respect to \textit{Alternating Algorithm} neighborhood, we expect the solution values to be similar. This can be seen in the table. A main observation here is that the \textit{2ex} heuristic does not combine well with \textit{AA}. Increased running time for both \textit{2ex+AA} and \textit{2exAAStep} confirms that \textit{AA} is more efficient in searching its much larger neighborhood.

\medskip

We then explored the effect of combining \textit{2exOptFirst} and \textit{AA}.
An algorithm that first runs \textit{AA} once and then applies \textit{2exOptFirst} until convergence will be referred to as \textbf{\textit{AA+2exOptFirst}}.
A more desirable variable neighborhood search based on the discussed heuristics will use the fact that most of the time running \textit{AA} until convergence is faster than even a single update of the solutions during the \textit{2exOptFirst} run. The algorithm \textit{\textbf{AA2exOptFirstStep}} (Algorithm \ref{AA2exOptFirstStep}) will use \textit{AA} to reach its local optimum and then will try to escape it by applying a single first possible improvement of the slower search \textit{2exOptFirst}. If successful, the process will start from the beginning with \textit{AA}. We will also add to the comparison variation with \textit{best improvement rule}, namely \textit{\textbf{AA2exOptStep}}.

\begin{algorithm}
\caption{$AA2exOptFirstStep$}
\label{AA2exOptFirstStep}
\begin{algorithmic}[0]\scriptsize
\Input integers $m, n$; $m \times m \times n \times n$ array $Q$; feasible solution $(\mb{x}, \mb{y})$ to given BAP
\Output feasible solution to given BAP

\While{True}

\State $(\mb{x}, \mb{y}) \gets AA(m, n, Q, (\mb{x}, \mb{y}))$ \Comment{running \textit{Alternating Algorithm} (Section \ref{sec:hex})}

\ForAll{$i_1 \in M$ \textbf{and all} $i_2 \in M \setminus \{i_1\}$}
\State $j_1 \gets$ assigned index to $i_1$ in $\mb{x}$
\State $j_2 \gets$ assigned index to $i_2$ in $\mb{x}$
\medskip

\State $\mb{x}^* \gets \mb{x}$
\State $x^*_{i_1j_1} \gets 0; \, x^*_{i_2j_2} \gets 0; \, x^*_{i_1j_2} \gets 1; \, x^*_{i_2j_1} \gets 1$ \Comment{applying 2-exchange}

\State $g_{kl} \gets \sum_{i, j \in M} q_{ijkl} x^*_{ij} \, \forall k, l \in N$
\State $\mb{y}^* \gets arg\,min_{\mb{y}' \in \yy} \sum_{k, l \in N} g_{kl} y'_{kl}$ \Comment{solving assignment problem for $\mb{y}$}

\If{$f(\mb{x}^*, \mb{y}^*) < f(\mb{x}, \mb{y})$ }
\State $\mb{x} \gets \mb{x}^*; \, \mb{y} \gets \mb{y}^*$
\State \textbf{continue while} \Comment{restarting the procedure \textbf{while} loop}
\EndIf

\EndFor

\ForAll{$k_1 \in N$ \textbf{and all} $k_2 \in N \setminus \{k_1\}$}
\State $l_1 \gets$ assigned index to $k_1$ in $\mb{y}$
\State $l_2 \gets$ assigned index to $k_2$ in $\mb{y}$
\medskip

\State $\mb{y}^* \gets \mb{y}$
\State $y^*_{k_1l_1} \gets 0; \, y^*_{k_2l_2} \gets 0; \, y^*_{k_1l_2} \gets 1; \, y^*_{k_2l_1} \gets 1$ \Comment{applying 2-exchange}

\State $e_{ij} \gets \sum_{k, l \in N} q_{ijkl} y^*_{kl} \, \forall i, j \in M$
\State $\mb{x}^* \gets arg\,min_{\mb{x}' \in \xx} \sum_{i, j \in M} e_{ij} x'_{ij}$ \Comment{solving assignment problem for $\mb{x}$}

\If{$f(\mb{x}^*, \mb{y}^*) < f(\mb{x}, \mb{y})$ }
\State $\mb{x} \gets \mb{x}^*; \, \mb{y} \gets \mb{y}^*$
\State \textbf{continue while} \Comment{restarting the procedure \textbf{while} loop}
\EndIf

\EndFor

\textbf{break} \Comment{algorithm converged, terminate}

\EndWhile

\State \textbf{return} ($\mb{x}$, $\mb{y}$)
\end{algorithmic}
\end{algorithm}

The results of these experiments are reported in Table \ref{vnsconvt2}. Here, we also report the number of iterations for \textit{AA2exOptStep} and \textit{AA2exOptFirstStep}, which represents the number of switches from the \textit{Alternating Algorithm} neighborhood to optimized $2$-exchange neighborhood before the algorithms converge.

\input{tables//vnsconv2.tex}

We have noticed that incorporating \textit{Alternating Algorithm} into \textit{optimized 2-exchange} yields a much better performance, bringing the convergence time down by at least an order of magnitude. Among variations, \textit{AA2exOptFirstStep} is consistently faster for \textit{uniform} and \textit{normal} instances. However, for \textit{euclidean} instances performance of all variable neighborhood search algorithms is similar. In fact, for euclidean instances of all sizes the average number of switches between neighborhoods is $1$, which implies that there is no possible improvement from the optimized $2$-exchange neighborhood after the Alternating Algorithm has converged. Thus, the special structure of instances must be always considered when developing metaheuristics for BAP.

Results on convergence time for all described algorithms from this sub-section, for \textit{uniform} instances, are given in Figure \ref{vnsconvut}.

\begin{figure}[h!]
\centering
\includegraphics[width=0.99\textwidth]{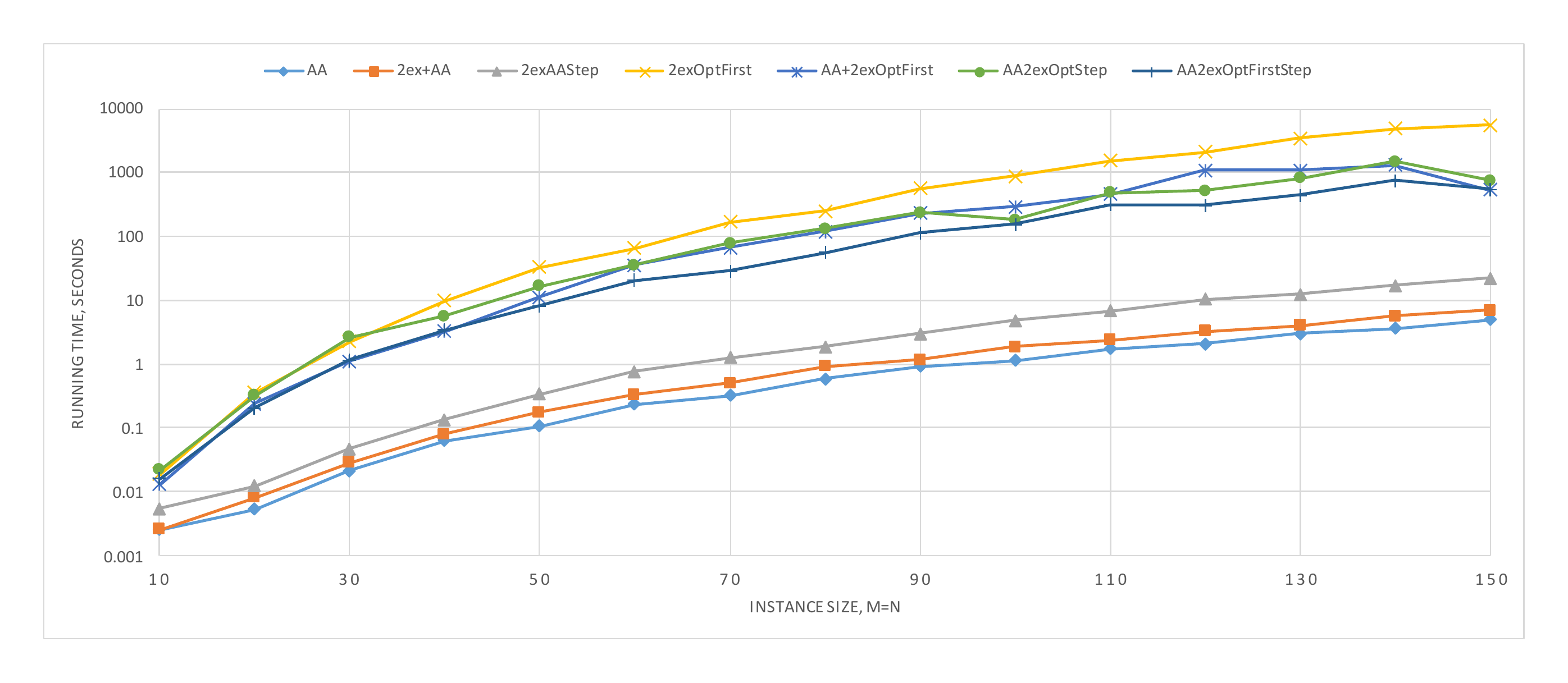}
\caption{Running time to reach the local optima by algorithms; \textit{uniform} instances}
\label{vnsconvut}
\end{figure}

\medskip

Our concluding set of experiments is dedicated to finding the most efficient combination of variable neighborhood search strategies and construction heuristics. We consider a variation of the VNS approach with the best convergence speed performance - \textit{AA2exOptFirstStep}. Namely, let \textit{\textbf{h-AA2exOptFirstStep}} be the algorithm that first generates $h$ starting solution, using \textit{RandomXYGreedy} strategy. It then proceeds to apply \textit{AA} to each of these solutions, selecting the best one and discarding the rest. After that \textit{h-AA2exOptFirstStep} will follow the description of \textit{AA2exOptFirstStep} (Algorithm \ref{AA2exOptFirstStep}) and will alternate between finding an improving solution using optimized $2$-exchange neighborhood and applying \textit{AA}, until the convergence to local optima. In this sense, \textit{AA2exOptFirstStep} and \textit{1-AA2exOptFirstStep} are equivalent implementations.

The single iteration of \textit{AA} requires $O(n^3)$ running time, whereas, a full exploration of the optimized $2$-exchange neighborhood will take $O(m^2n^3)$. From the experiments in Section \ref{sec:expls} we also know that it usually takes \textit{AA} less than 10 iterations to converge. Based on these observations, for the following experimental analysis we have chosen $h$ for \textit{h-AA2exOptFirstStep} as $h \in \{4, 10, 100\}$.

In addition to versions of \textit{h-AA2exOptFirstStep} we consider a simple multi-start \textit{AA} strategy that performed well in previous experiments (see Section \ref{sec:explsms}), denoted \textit{\textbf{msAA}}. Now however, the starting solution each time is generated using \textit{RandomXYGreedy} construction heuristic. As the time limit for this multi-start approach we select the highest convergence time among all \textit{h-AA2exOptFirstStep} variations. As it often happens during the time-limited multi-start procedures, the best solution will be found before the final iteration. Hence, in addition to the total number we also report the average iteration (\textit{best iter}) at which the finally reported solution was found, and the standard deviation of this value.

See the results of these experiments in Table \ref{vnsmst} and Figure \ref{vnsmsuv}.

\input{tables//vnsms.tex}

\begin{figure}[h!]
\centering
\includegraphics[width=0.99\textwidth]{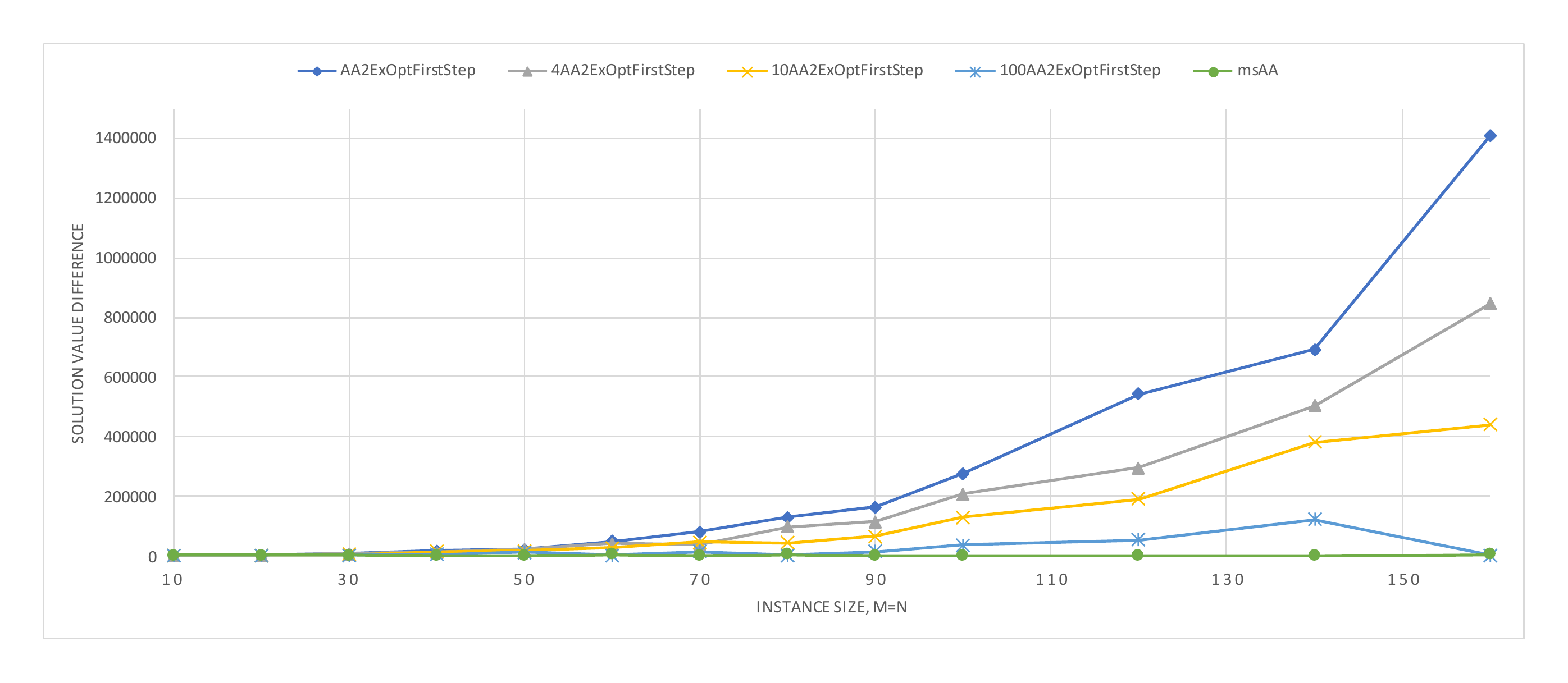}
\caption{Difference between solution values (to the best) for algorithms; \textit{uniform} instances}
\label{vnsmsuv}
\end{figure}

Under this considerations, multi-start \textit{AA} once again performed the best. \textit{h-AA2exOptFirstStep} variations were the more efficient, the higher the number $h$ was. Interestingly, for several instance sizes, the average iteration of finding the best solution by \textit{msAA} is substantially bellow $100$. However, the observed standard deviation is very high, which hints towards the variability of the solutions produced by \textit{AA}. To confirm this, we present in Figures \ref{100AAu}, \ref{100AAn} and \ref{100AAe} the spread of solution values produced by applying \textit{AA} to the solution of \textit{RandomXYGreedy} (denoted as \textit{RandomXYGreedy+AA}). All three instances in these charts are of size $m=n=100$, and we perform $100$ runs of this metaheuristic.

\begin{figure}[h!]
\centering
\includegraphics[width=0.8\textwidth]{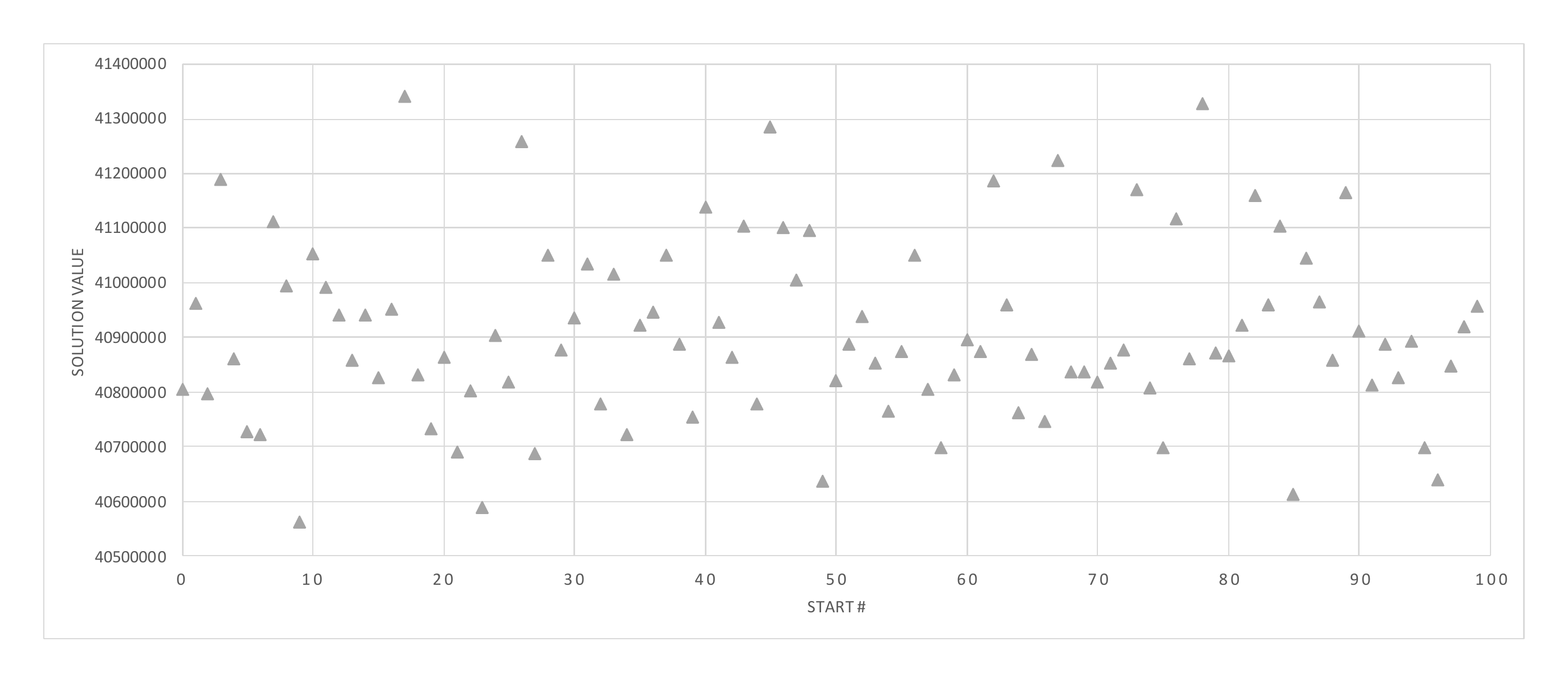}
\caption{Objective solution values for \textit{RandomXYGreedy+AA} metaheuristic; \textit{uniform} $100 \times 100$ instance}
\label{100AAu}
\end{figure}

\begin{figure}[h!]
\centering
\includegraphics[width=0.8\textwidth]{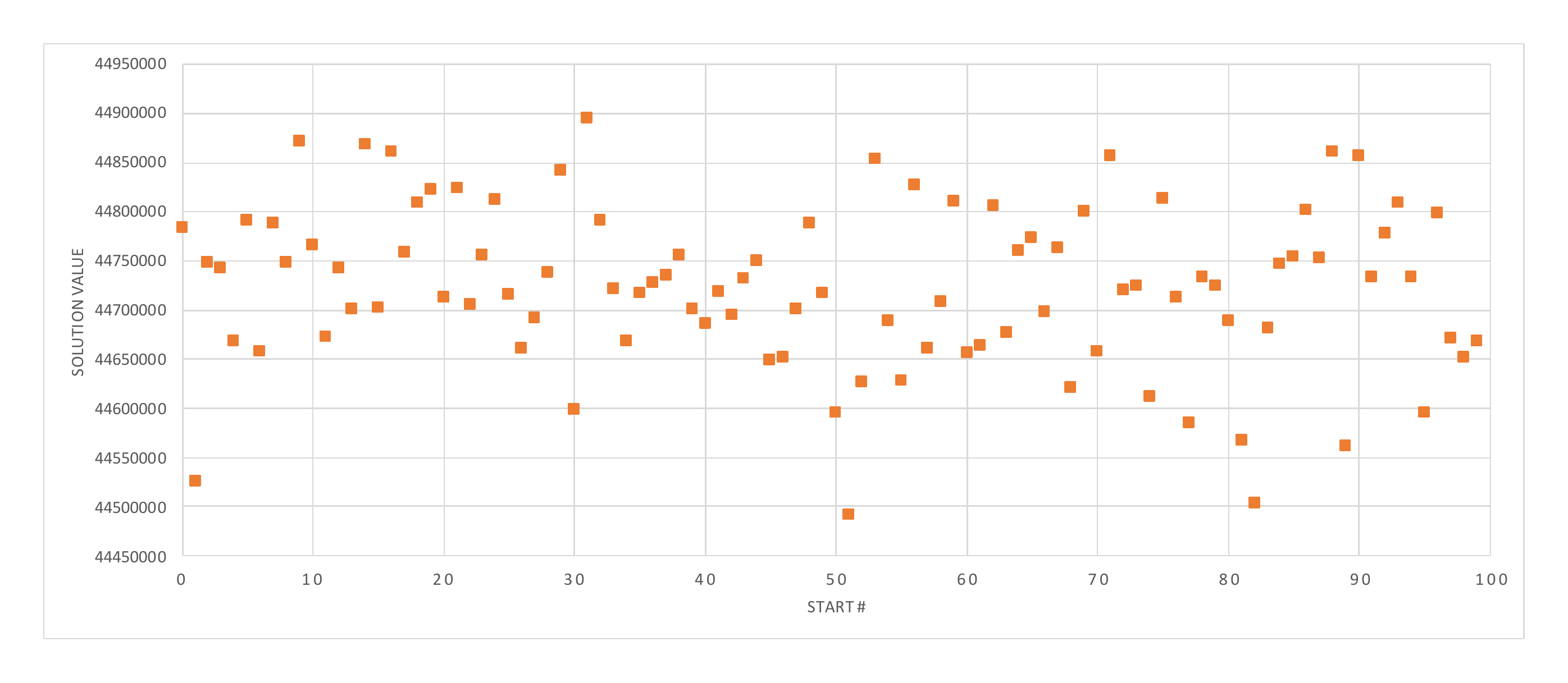}
\caption{Objective solution values for \textit{RandomXYGreedy+AA} metaheuristic; \textit{normal} $100 \times 100$ instance}
\label{100AAn}
\end{figure}

\begin{figure}[h!]
\centering
\includegraphics[width=0.8\textwidth]{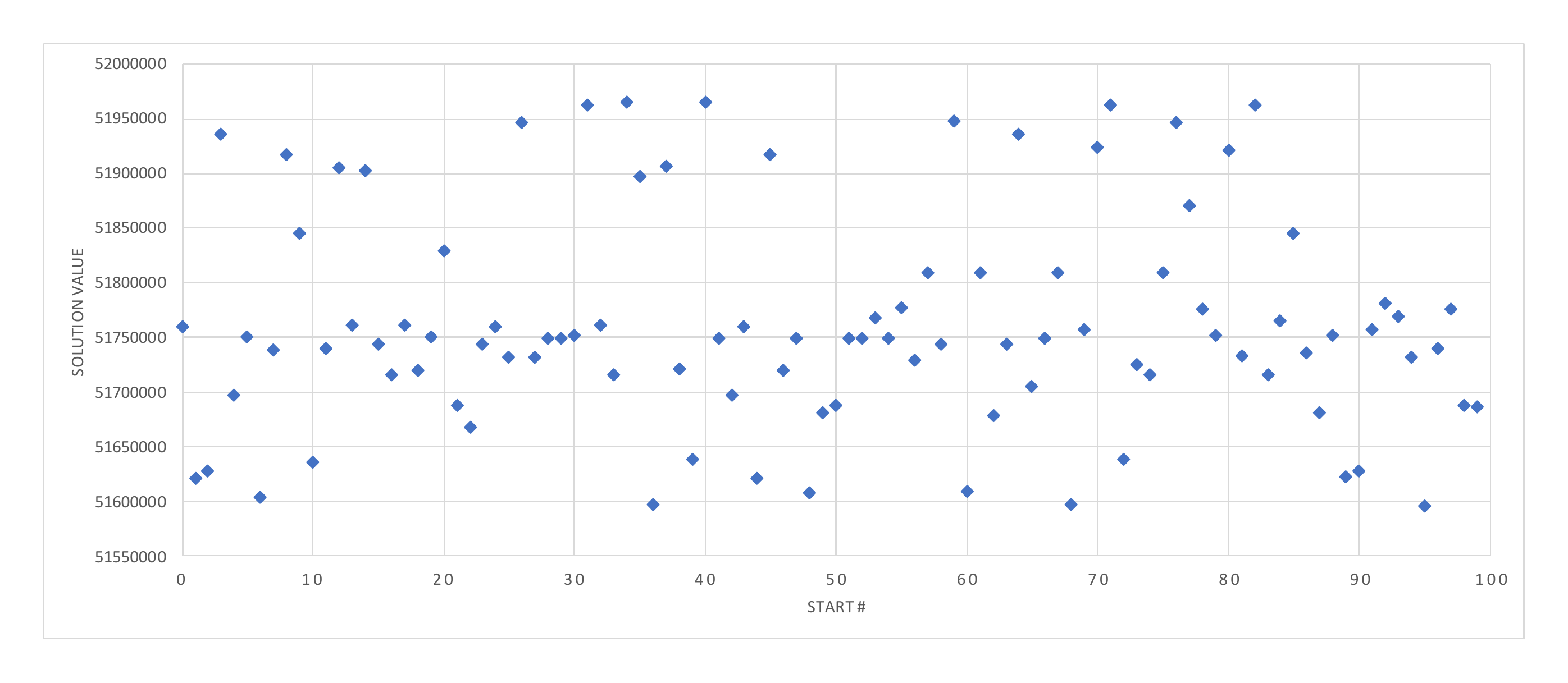}
\caption{Objective solution values for \textit{RandomXYGreedy+AA} metaheuristic; \textit{euclidean} $100 \times 100$ instance}
\label{100AAe}
\end{figure}

At this point, we conclude that optimized $2$-exchange neighborhood is too costly to explore, in comparison to the neighborhood that \textit{AA} is based on. For the general case it is more effective to do several more restarts of \textit{AA} from \textit{RandomXYGreedy} solutions then to spend time escaping local optima with even a single step of \textit{2exOpt}. It is suggested to only use efficient implementations of VNS that explore optimized $2$-exchange neighborhood as the final step of any metaheuristic. In this way you can improve your solution quality without excessive time spending, while leaving all the heavy work for \textit{Alternation Algorithm}.

Our previous experiments that involve multi-start strategies (in this section and Section \ref{sec:explsms}) have reasonable time limit restrictions. This considerations are important when developing algorithms to run on real-life instances. However, we are also interested in behavior of multi-start AA and multi-start VNS in the case of unlimited (or unreasonably large) running time constraints. Figure \ref{ut100u} presents results of running multi-start \textit{AA}, multi-start \textit{1-AA2exOptFirstStep} and multi-start \textit{100-AA2exOptFirstStep}, for a single $100 \times 100$ \textit{uniform} instance, for an exceedingly long period of time. All starts are made from the solutions generated by \textit{RandomXYGreedy} heuristic. Here we report the change of the best found solution value, depending on time.

\begin{figure}[h!]
\centering
\includegraphics[width=0.99\textwidth]{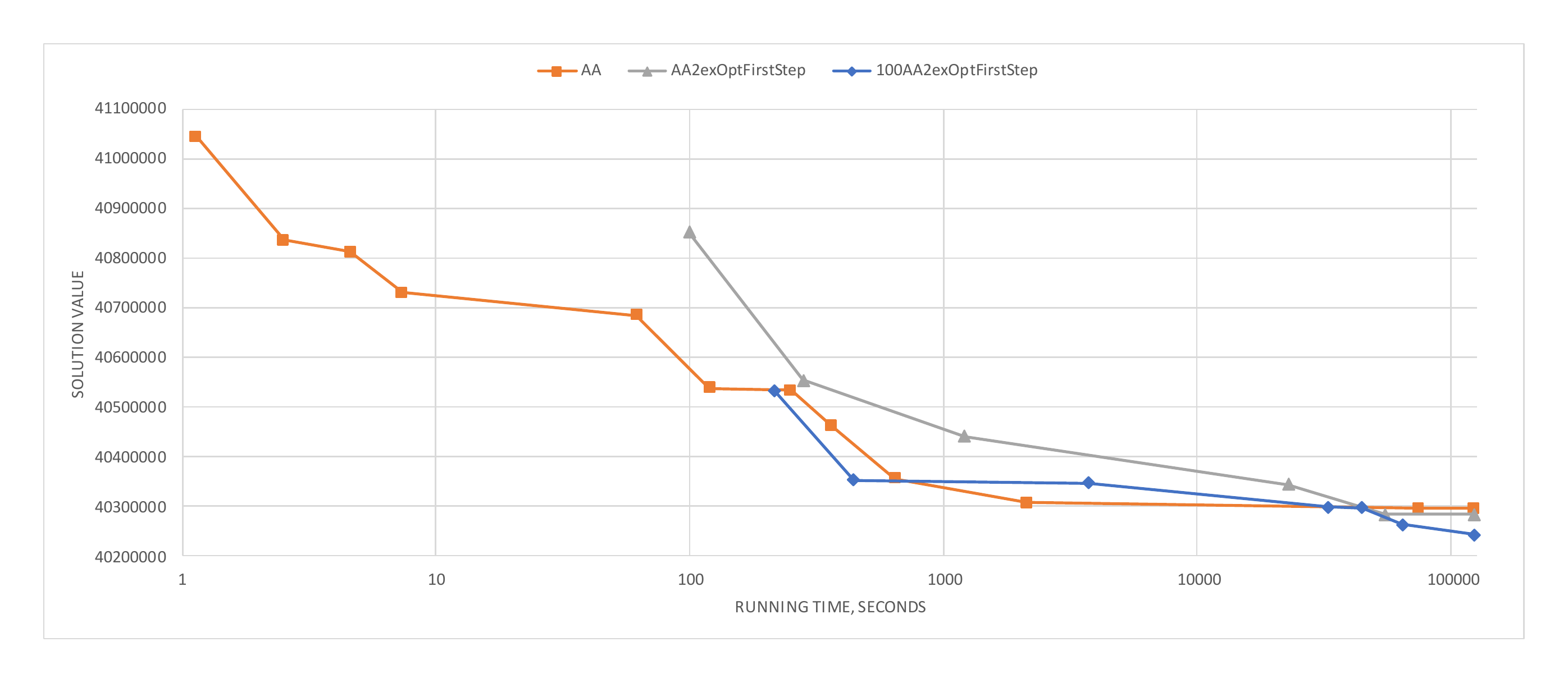}
\caption{Improvement over time of best found objective solution value for multi-start heuristics; \textit{uniform} $100 \times 100$ instance}
\label{ut100u}
\end{figure}

We can see that after 50000 seconds (0.58 days of running) multi-start VNS strategies begin to dominate the multi-start \textit{AA}, even though the later approach is much more efficient in solution space exploration for short running times. This observation is consistent with optimized $h$-exchange being a more powerful neighborhood in terms of solutions quality.

%%%%%%%%%%%%%%%%%%%%%%%%%%%%%%%%%%%%%%%%%%%%%%%%%%%%%%%%%%%%%%%%
%%%%%%%%%%%%%%%%%%%%%%%%%%%%%%%%%%%%%%%%%%%%%%%%%%%%%%%%%%%%%%%%
%%%%%%%%%%%%%%%%%%%%%%%%%%%%%%%%%%%%%%%%%%%%%%%%%%%%%%%%%%%%%%%%
%%%%%%%%%%%%%%%%%%%%%%%%%%%%%%%%%%%%%%%%%%%%%%%%%%%%%%%%%%%%%%%%
\section{Conclusion}
\label{sec:conclusion}

We have presented the first systematic experimental analysis of heuristics for BAP along with some theoretical results on local search algorithms worst case performance.

\smallskip

Three classes of neighborhoods - $h$-exchange, $[h,p]$-exchange and shift based - are introduced. Some of the neighborhoods are of an exponential size but can be searched for an improving solution in polynomial time. Analysis of local optimums in terms of domination properties and relation to average value $\Aa(Q,C,D)$ are presented.

Several greedy, semi-greedy and rounding construction heuristics are proposed for generating reasonable quality solution quickly. Experimental results show that \textit{RandomXYGreedy} is a good alternative among the approaches. The built-in randomized decision steps make this heuristic valuable for generating starting solutions for improvement algorithms within a multistart framework.

Extensive computational analysis has been carried out on the searches based on described neighborhoods. The experimental results suggest that the very large-scale neighborhood (VLSN) search algorithm - \textit{Alternating Algorithm (AA)}, when used within multi-start framework, yields a more balanced heuristic in terms of running time and solution quality. A variable neighborhood search (VNS) algorithm, that strategically uses optimized $2$-exchange neighborhood and \textit{AA} neighborhood, produced superior outcomes. However, this came with the downside of a significantly larger computational time.

\smallskip

We hope that this study inspires additional research work on the bilinear assignment model, particularly in the area of design and analysis of exact and heuristic algorithms.
%To assist researchers who pursue experimental study on this versatile optimization model, we provided test instances, instance generators and additional details of our experimental study in a supplementary page [?].

%%%%%%%%%%%%%%%%%%%%%%%%%%%%%%%%%%%%%%%%%%%%%%%%%%%%%%%%%%%%%%%%
%%%%%%%%%%%%%%%%%%%%%%%%%%%%%%%%%%%%%%%%%%%%%%%%%%%%%%%%%%%%%%%%
%%%%%%%%%%%%%%%%%%%%%%%%%%%%%%%%%%%%%%%%%%%%%%%%%%%%%%%%%%%%%%%%
%%%%%%%%%%%%%%%%%%%%%%%%%%%%%%%%%%%%%%%%%%%%%%%%%%%%%%%%%%%%%%%%
\section*{Acknowledgment}

This work was supported by NSERC Discovery Grant and NSERC Accelerator Supplement awarded to Abraham P Punnen as well as NSERC Discovery Grant awarded to Binay Bhattacharya.

%%%%%%%%%%%%%%%%%%%%%%%%%%%%%%%%%%%%%%%%%%%%%%%%%%%%%%%%%%%%%%%%
%%%%%%%%%%%%%%%%%%%%%%%%%%%%%%%%%%%%%%%%%%%%%%%%%%%%%%%%%%%%%%%%
%%%%%%%%%%%%%%%%%%%%%%%%%%%%%%%%%%%%%%%%%%%%%%%%%%%%%%%%%%%%%%%%
%%%%%%%%%%%%%%%%%%%%%%%%%%%%%%%%%%%%%%%%%%%%%%%%%%%%%%%%%%%%%%%%
\bibliographystyle{plain}
\bibliography{references}

\end{document}

%% file: tables/constr.tex
\clearpage
\newgeometry{margin=1cm}
\thispagestyle{empty}

\begin{sidewaystable*}[t!]
\centering
\small
\caption{Solution value and running time in seconds for construction heuristics}
\label{constrt}
\scalebox{0.9}{\begin{tabular}{@{}lrrcrcrcrcrc@{}}
\toprule
 & & \multicolumn{2}{c}{RandomXYGreedy} & \multicolumn{2}{c}{Greedy} & \multicolumn{2}{c}{GreedyRandomized2} & \multicolumn{2}{c}{GreedyRandomized4} & \multicolumn{2}{c}{Rounding}\\ \cmidrule(lr){3-4} \cmidrule(lr){5-6} \cmidrule(lr){7-8} \cmidrule(lr){9-10} \cmidrule(lr){11-12}
instances & $\Aa$ & value & time & value & time & value & time & value & time & value & time\\
\midrule
uniform 20x20 & 79975 & 62981 & 0.0011 & 61930 & 0.0016 & 61824 & 0.0015 & 62997 & 0.0023 & \textbf{58587} & 0.0282 \\
uniform 40x40 & 1280013 & 1039365 & 0.0024 & 1038410 & 0.0085 & 1046862 & 0.0117 & 1047444 & 0.0107 & \textbf{1005375} & 0.4083 \\
uniform 60x60 & 6480224 & 5335157 & 0.0057 & 5399004 & 0.0362 & 5430190 & 0.0403 & 5429077 & 0.0381 & \textbf{5311287} & 2.076 \\
uniform 80x80 & 20480398 & 17179410 & 0.0119 & 17393975 & 0.0901 & 17427649 & 0.1092 & 17455112 & 0.1231 & \textbf{17127745} & 8.6041 \\
uniform 100x100 & 50001181 & \textbf{42492213} & 0.0205 & 43134618 & 0.1797 & 43115743 & 0.1755 & 43209207 & 0.2431 & 42521606 & 29.3038 \\
uniform 120x120 & 103680291 & \textbf{88710617} & 0.0334 & 90317432 & 0.2459 & 90450040 & 0.3127 & 90388890 & 0.3208 & 89342939 & 90.1245 \\
uniform 140x140 & 192079012 & \textbf{165656443} & 0.0518 & 168664018 & 0.404 & 168695610 & 0.5922 & 168683177 & 0.5869 & 166927409 & 196.3766 \\
uniform 160x160 & 327679690 & \textbf{284623314} & 0.0768 & 289819325 & 0.939 & 289847112 & 0.9922 & 290034508 & 0.9862 & 287148038 & 339.6329 \\
uniform 180x180 & 524879096 & \textbf{458395075} & 0.1088 & 466419210 & 1.0135 & 466652862 & 1.107 & 466938203 & 1.5316 & 462852252 & 539.6931 \\
normal 20x20 & 79977 & 69989 & 0.0011 & 69032 & 0.0013 & 69322 & 0.0015 & 69899 & 0.0022 & \textbf{67367} & 0.0275 \\
normal 40x40 & 1280007 & 1137550 & 0.0022 & 1137478 & 0.008 & 1139150 & 0.0098 & 1139608 & 0.0116 & \textbf{1123670} & 0.3902 \\
normal 60x60 & 6480142 & 5825775 & 0.0055 & 5847641 & 0.0229 & 5841178 & 0.0277 & 5860741 & 0.0427 & \textbf{5795676} & 2.0257 \\
normal 80x80 & 20480028 & 18555962 & 0.0108 & 18696934 & 0.0613 & 18658585 & 0.0772 & 18697475 & 0.102 & \textbf{18544051} & 6.9208 \\
normal 100x100 & 50000062 & 45647505 & 0.02 & 45909621 & 0.1293 & 45925799 & 0.1584 & 45943220 & 0.1958 & \textbf{45643447} & 30.2969 \\
normal 120x120 & 103680643 & \textbf{94952757} & 0.0325 & 95765991 & 0.2465 & 95711199 & 0.2967 & 95757531 & 0.3385 & 95332171 & 80.9744 \\
normal 140x140 & 192079732 & \textbf{176656351} & 0.0507 & 178279212 & 0.4034 & 178238835 & 0.4936 & 178233293 & 0.556 & 177501940 & 179.0639 \\
normal 160x160 & 327681533 & \textbf{302496650} & 0.0738 & 305379404 & 0.746 & 305333912 & 0.696 & 305345983 & 0.823 & 304080792 & 310.9162 \\
normal 180x180 & 524880349 & \textbf{486132477} & 0.1056 & 490345723 & 0.8888 & 490464093 & 1.0742 & 490656416 & 1.3211 & 489077716 & 540.4644 \\
euclidean 20x20 & 95297 & 93756 & 0.0011 & 98864 & 0.0013 & 99027 & 0.0014 & 98104 & 0.0015 & \textbf{85564} & 0.0276 \\
euclidean 40x40 & 1554313 & 1540492 & 0.0024 & 1559829 & 0.0111 & 1546894 & 0.0116 & 1551881 & 0.0123 & \textbf{1430068} & 0.4218 \\
euclidean 60x60 & 8003105 & 7821082 & 0.0063 & 8021089 & 0.0445 & 8014594 & 0.0461 & 7945751 & 0.0489 & \textbf{7331236} & 1.9805 \\
euclidean 80x80 & 24906273 & 24190227 & 0.0129 & 24873255 & 0.0611 & 24799662 & 0.0954 & 24853670 & 0.0805 & \textbf{23145446} & 6.141 \\
euclidean 100x100 & 61053265 & 59345477 & 0.0235 & 60305521 & 0.103 & 59882626 & 0.1285 & 60052837 & 0.1223 & \textbf{56848260} & 31.8484 \\
euclidean 120x120 & 126198999 & 121816738 & 0.0389 & 123601338 & 0.2986 & 123829252 & 0.305 & 124053452 & 0.3252 & \textbf{117754675} & 93.6024 \\
euclidean 140x140 & 230673448 & 221785417 & 0.0617 & 227949036 & 0.4082 & 227508295 & 0.4637 & 227854403 & 0.4979 & \textbf{214876628} & 183.0906 \\
euclidean 160x160 & 404912898 & 390412111 & 0.0897 & 395260253 & 0.8908 & 398388924 & 0.8284 & 396277525 & 1.0551 & \textbf{378608021} & 309.2262 \\
euclidean 180x180 & 635700756 & 607470603 & 0.1289 & 623035384 & 1.1913 & 625456121 & 1.356 & 623393649 & 1.4349 & \textbf{593800828} & 548.8153 \\
\bottomrule
\end{tabular}}
\end{sidewaystable*}
\clearpage
\restoregeometry

%% file: tables/summary.tex
\begin{table}[!htb]
\centering
\caption{Asymptotic running time and neighborhood size per iteration for local searches}
\label{sum}
\scalebox{1.0}{\begin{tabular}{@{}ccc@{}}
\toprule
name & running time & neighborhood size per iteration\\
\midrule
2ex & $O(n^3 + I n^2)$ & $\Theta(n^2)$\\
Shift & $O(I n^3)$ & $n$\\
ShiftShuffle & $O(I n^3)$ & $\Theta(n^2)$\\
3ex & $O(I n^3)$ & $\Theta(n^3)$\\
AA & $O(I n^3)$ & $n!$\\
DualShift & $O(n^4)$ & $n^2$\\
Dual2ex & $O(I n^4)$ & $\Theta(n^4)$\\
ShiftOpt & $O(I n^4)$ & $n \cdot n!$\\
2exOptHeuristic & $O(I n^4)$ & $\Theta(n^2 \cdot n!)^*$\\
2exOpt & $O(I n^5)$ & $\Theta(n^2 \cdot n!)$\\
3exOpt & $O(I n^6)$ & $\Theta(n^3 \cdot n!)$\\
\bottomrule
\multicolumn{3}{@{}l@{}}{\scalebox{.8}{* 2exOptHeuristic does not fully explore the neighborhood.}}
\end{tabular}}
\end{table}

%% file: tables/conv.tex
\clearpage
\newgeometry{margin=1cm}
\thispagestyle{empty}

\begin{sidewaystable*}[t!]
\centering
\small
\caption{Solution value, running time in seconds and number of iterations for local searches}
\label{convt}\scalebox{0.85}{\begin{tabular}{@{}lrrlcrlcrlcrlcrlc@{}}
\toprule
 & & \multicolumn{3}{c}{2ex} & \multicolumn{3}{c}{3ex} & \multicolumn{3}{c}{AA} & \multicolumn{3}{c}{2exOpt} & \multicolumn{3}{c}{2exOptFirst}\\ \cmidrule(lr){3-5} \cmidrule(lr){6-8} \cmidrule(lr){9-11} \cmidrule(lr){12-14} \cmidrule(lr){15-17}
instances & $\Aa$ & value & time & iter & value & time & iter & value & time & iter & value & time & iter & value & time & iter\\
\midrule
uniform 10x10 & 4995 & 3378 & 0.0 & 9 & 3241 & 0.0 & 9 & 3385 & 0.0 & 3 & \textbf{3103} & 0.04 & 4 & 3128 & 0.02 & 11 \\
uniform 20x20 & 80043 & 59371 & 0.0 & 20 & 56593 & 0.01 & 18 & 56097 & 0.01 & 4 & \textbf{54912} & 0.68 & 6 & 55059 & 0.34 & 25 \\
uniform 30x30 & 404944 & 310455 & 0.02 & 32 & 297569 & 0.05 & 28 & 298787 & 0.02 & 4 & 291520 & 3.96 & 6 & \textbf{291268} & 3.09 & 46 \\
uniform 40x40 & 1279785 & 1003731 & 0.04 & 45 & 977498 & 0.14 & 39 & 971400 & 0.06 & 5 & \textbf{954676} & 21.71 & 10 & 957381 & 8.46 & 56 \\
uniform 50x50 & 3124809 & 2493822 & 0.08 & 57 & 2433665 & 0.32 & 49 & 2416832 & 0.13 & 5 & \textbf{2385232} & 63.49 & 11 & 2389496 & 24.94 & 73 \\
uniform 60x60 & 6479878 & 5256357 & 0.15 & 74 & 5149634 & 0.59 & 55 & 5098653 & 0.26 & 6 & 5056566 & 143.48 & 11 & \textbf{5031368} & 80.32 & 97 \\
uniform 70x70 & 12005619 & 9844646 & 0.24 & 85 & 9682798 & 1.1 & 67 & 9587489 & 0.38 & 6 & \textbf{9469736} & 326.04 & 14 & 9472549 & 156.78 & 114 \\
uniform 80x80 & 20480209 & 17022523 & 0.37 & 96 & 16694088 & 1.81 & 75 & 16519908 & 0.66 & 7 & 16388545 & 504.34 & 12 & \textbf{16355658} & 285.23 & 136 \\
uniform 90x90 & 32803918 & 27479017 & 0.52 & 111 & 26978715 & 2.97 & 88 & 26650508 & 1.08 & 8 & 26563051 & 882.81 & 13 & \textbf{26514860} & 497.74 & 158 \\
uniform 100x100 & 49999078 & 42138227 & 0.74 & 124 & 41363121 & 4.96 & 109 & 41031842 & 1.45 & 8 & 40912367 & 1480.03 & 14 & \textbf{40767754} & 864.39 & 172 \\
uniform 110x110 & 73206906 & 61988038 & 1.06 & 148 & 61179121 & 6.57 & 109 & 60529975 & 1.92 & 7 & 60162728 & 2406.29 & 15 & \textbf{60068824} & 1504.27 & 196 \\
uniform 120x120 & 103679901 & 88602187 & 1.23 & 137 & 87330165 & 8.52 & 109 & 86174642 & 2.61 & 8 & 85872203 & 3865.67 & 18 & \textbf{85670906} & 1917.76 & 201 \\
normal 10x10 & 4999 & 4044 & 0.0 & 10 & 4019 & 0.0 & 9 & 4040 & 0.0 & 2 & 3910 & 0.03 & 4 & \textbf{3862} & 0.02 & 14 \\
normal 20x20 & 79955 & 67321 & 0.0 & 20 & 66520 & 0.01 & 16 & 66179 & 0.01 & 3 & \textbf{64913} & 0.79 & 7 & 65363 & 0.33 & 25 \\
normal 30x30 & 404959 & 348058 & 0.02 & 34 & 342238 & 0.06 & 29 & 343639 & 0.03 & 4 & \textbf{338796} & 4.98 & 8 & 339162 & 2.61 & 45 \\
normal 40x40 & 1279974 & 1119684 & 0.04 & 46 & 1111127 & 0.14 & 33 & 1099106 & 0.07 & 6 & 1089996 & 23.21 & 10 & \textbf{1089752} & 10.61 & 60 \\
normal 50x50 & 3124879 & 2752326 & 0.08 & 63 & 2737137 & 0.34 & 43 & 2711191 & 0.14 & 6 & 2696287 & 65.48 & 11 & \textbf{2696062} & 32.57 & 77 \\
normal 60x60 & 6479794 & 5769522 & 0.16 & 73 & 5707107 & 0.7 & 53 & 5665027 & 0.3 & 7 & 5640412 & 151.84 & 12 & \textbf{5633463} & 81.97 & 99 \\
normal 70x70 & 12004939 & 10738678 & 0.24 & 88 & 10641129 & 1.3 & 65 & 10596245 & 0.42 & 6 & 10544640 & 316.24 & 13 & \textbf{10538513} & 144.42 & 116 \\
normal 80x80 & 20480106 & 18434378 & 0.38 & 103 & 18282395 & 2.35 & 80 & 18173927 & 0.71 & 7 & 18126933 & 537.29 & 12 & \textbf{18095224} & 338.76 & 132 \\
normal 90x90 & 32805972 & 29736595 & 0.51 & 108 & 29408513 & 3.79 & 91 & 29245481 & 0.92 & 6 & 29176212 & 1017.08 & 14 & \textbf{29165974} & 500.62 & 151 \\
normal 100x100 & 49999105 & 45514117 & 0.71 & 122 & 45009249 & 5.69 & 100 & 44798388 & 1.45 & 7 & 44635991 & 1602.09 & 15 & \textbf{44603238} & 940.19 & 176 \\
normal 110x110 & 73205050 & 66768499 & 1.01 & 142 & 66224593 & 8.26 & 110 & 65812495 & 2.69 & 10 & 65716978 & 2218.71 & 13 & \textbf{65539744} & 1632.32 & 193 \\
normal 120x120 & 103681336 & 95001950 & 1.32 & 147 & 94151507 & 11.24 & 116 & 93702171 & 2.16 & 6 & 93322807 & 4645.28 & 20 & \textbf{93248160} & 2130.64 & 215 \\
euclidean 10x10 & 6186 & 5397 & 0.0 & 13 & 5379 & 0.0 & 12 & 5404 & 0.0 & 3 & \textbf{5368} & 0.05 & 4 & 5375 & 0.03 & 16 \\
euclidean 20x20 & 95834 & 82325 & 0.01 & 41 & 82293 & 0.01 & 25 & 82242 & 0.01 & 3 & 82160 & 1.27 & 5 & \textbf{81813} & 1.52 & 49 \\
euclidean 30x30 & 490614 & 419174 & 0.02 & 61 & 418942 & 0.07 & 40 & 419000 & 0.03 & 3 & \textbf{416436} & 9.13 & 5 & 417339 & 18.19 & 98 \\
euclidean 40x40 & 1553544 & 1314659 & 0.07 & 87 & 1312649 & 0.21 & 59 & 1311131 & 0.07 & 3 & \textbf{1309701} & 37.08 & 5 & 1311093 & 90.91 & 156 \\
euclidean 50x50 & 3761359 & 3178424 & 0.14 & 112 & 3173915 & 0.5 & 78 & 3178006 & 0.16 & 4 & \textbf{3167772} & 134.77 & 7 & 3168388 & 314.91 & 211 \\
euclidean 60x60 & 7999029 & 6740779 & 0.26 & 141 & 6720560 & 1.04 & 98 & \textbf{6714400} & 0.23 & 4 & 6714689 & 314.14 & 7 & 6716877 & 1012.93 & 296 \\
euclidean 70x70 & 14909550 & 12533959 & 0.45 & 180 & 12500249 & 1.92 & 117 & 12490034 & 0.42 & 4 & \textbf{12487021} & 674.66 & 7 & 12499281 & 2354.68 & 366 \\
euclidean 80x80 & 25210773 & 21188706 & 0.68 & 200 & 21182227 & 3.2 & 133 & 21160309 & 0.55 & 4 & \textbf{21150070} & 1222.19 & 6 & 21156445 & 5250.01 & 456 \\
euclidean 90x90 & 39495474 & 33083033 & 1.04 & 240 & 33072079 & 4.87 & 145 & 33082326 & 0.98 & 5 & \textbf{33049474} & 2017.96 & 6 & 33089283 & 10482.75 & 556 \\
\bottomrule
\end{tabular}}
\end{sidewaystable*}
\clearpage
\restoregeometry

%% file: tables/ms.tex
\clearpage
\newgeometry{margin=1cm}
\thispagestyle{empty}

\begin{sidewaystable*}[t!]
\centering
\small
\caption{Solution value and number of starts for time-limited multi-start local searches}
\label{mst}\scalebox{0.9}{\begin{tabular}{@{}lcrcrcrcrcrcrc@{}}
\toprule
 & & \multicolumn{2}{c}{3exOptFirst} & \multicolumn{2}{c}{2exOpt} & \multicolumn{2}{c}{2exOptFirst} & \multicolumn{2}{c}{AA} & \multicolumn{2}{c}{2ex} & \multicolumn{2}{c}{2exFirst}\\ \cmidrule(lr){3-4} \cmidrule(lr){5-6} \cmidrule(lr){7-8} \cmidrule(lr){9-10} \cmidrule(lr){11-12} \cmidrule(lr){13-14}
instances & time limit & value & starts & value & starts & value & starts & value & starts & value & starts & value & starts\\
\midrule
uniform 10x10 & 0.1 & 3059 & 1 & 2943 & 3 & 2974 & 5 & 2946 & 97 & \textbf{2934} & 221 & 2980 & 176 \\
uniform 20x20 & 2.7 & 54250 & 1 & 53496 & 4 & 53286 & 8 & \textbf{53096} & 428 & 53983 & 997 & 54244 & 879 \\
uniform 30x30 & 23.4 & 290200 & 1 & 288401 & 5 & 285630 & 10 & \textbf{285271} & 919 & 292991 & 1859 & 292363 & 1695 \\
uniform 40x40 & 103.2 & 948029 & 1 & 943982 & 5 & 940718 & 10 & \textbf{936113} & 1528 & 963120 & 2858 & 960093 & 2679 \\
uniform 50x50 & 531.7 & 2370639 & 1 & 2365473 & 8 & 2358811 & 18 & \textbf{2346865} & 3664 & 2410678 & 6592 & 2401247 & 6337 \\
uniform 60x60 & 1148.5 & 5017422 & 1 & 5003247 & 7 & 4989212 & 16 & \textbf{4980930} & 4221 & 5105064 & 7747 & 5092544 & 7522 \\
uniform 70x70 & 3291.3 & 9429464 & 1 & 9421085 & 10 & 9404126 & 21 & \textbf{9369944} & 7017 & 9601891 & 13499 & 9583585 & 13009 \\
uniform 80x80 & 3763.3 & 16406602 & 1 & 16319588 & 7 & 16241213 & 13 & \textbf{16229861} & 5031 & 16612105 & 10017 & 16583987 & 9578 \\
normal 10x10 & 0.1 & 3857 & 1 & 3838 & 2 & 3851 & 5 & 3828 & 91 & \textbf{3818} & 208 & 3847 & 162 \\
normal 20x20 & 2.5 & 65014 & 1 & 64635 & 4 & 64433 & 7 & \textbf{64020} & 396 & 64867 & 902 & 64738 & 769 \\
normal 30x30 & 23.4 & 337626 & 1 & 336552 & 5 & 335378 & 10 & \textbf{335042} & 899 & 339448 & 1818 & 338849 & 1623 \\
normal 40x40 & 113.3 & 1086083 & 1 & 1082094 & 5 & 1081530 & 12 & \textbf{1078755} & 1675 & 1092923 & 3063 & 1091803 & 2840 \\
normal 50x50 & 469.3 & 2688595 & 1 & 2679334 & 8 & 2677720 & 16 & \textbf{2672481} & 3217 & 2711913 & 5807 & 2704948 & 5475 \\
normal 60x60 & 933.4 & 5640721 & 1 & 5627391 & 6 & 5612362 & 13 & \textbf{5604229} & 3413 & 5679037 & 6216 & 5672749 & 5979 \\
normal 70x70 & 3593.3 & 10512493 & 1 & 10492591 & 12 & 10483432 & 25 & \textbf{10474343} & 7685 & 10604646 & 14559 & 10591133 & 13903 \\
normal 80x80 & 11339.0 & \textbf{17989971} & 1 & 17993643 & 20 & 18010732 & 42 & 17995894 & 15435 & 18226724 & 29827 & 18209532 & 28425 \\
euclidean 10x10 & 0.1 & 5447 & 1 & 5430 & 3 & 5445 & 3 & \textbf{5427} & 98 & 5427 & 266 & 5427 & 162 \\
euclidean 20x20 & 5.1 & 82409 & 1 & 81717 & 4 & 81710 & 4 & \textbf{81573} & 589 & 81573 & 1283 & 81575 & 747 \\
euclidean 30x30 & 70.1 & 418658 & 1 & 415529 & 7 & 415419 & 4 & \textbf{414767} & 2399 & 414774 & 3382 & 414808 & 1732 \\
euclidean 40x40 & 390.3 & 1321385 & 1 & 1317439 & 9 & 1317948 & 4 & \textbf{1316409} & 5459 & 1316509 & 6197 & 1316771 & 3010 \\
euclidean 50x50 & 1675.4 & 3151591 & 1 & 3136628 & 13 & 3139866 & 4 & \textbf{3135362} & 11411 & 3135723 & 11993 & 3136122 & 5359 \\
euclidean 60x60 & 4604.9 & 6563921 & 1 & 6532789 & 15 & 6537657 & 4 & \textbf{6529495} & 17621 & 6530835 & 17448 & 6532247 & 6641 \\
\bottomrule
\end{tabular}}
\end{sidewaystable*}
\clearpage
\restoregeometry

%% file: tables/vnsconv1.tex
\begin{table}[!htb]
\centering
\caption{Solution value, running time in seconds and number of iterations for \textit{Alternating Algorithm} and variations (convergence to local optima)}
\label{vnsconvt1}\scalebox{0.8}{\begin{tabular}{@{}lrlrlrlc@{}}
\toprule
 & \multicolumn{2}{c}{AA} & \multicolumn{2}{c}{2ex+AA} & \multicolumn{3}{c}{2exAAStep}\\ \cmidrule(lr){2-3} \cmidrule(lr){4-5} \cmidrule(lr){6-8}
instances & value & time & value & time & value & time & iter\\
\midrule
uniform 10x10 & 3255 & \textbf{0.0} & 3305 & 0.0 & 3322 & 0.01 & 1 \\
uniform 20x20 & 56287 & \textbf{0.01} & 56136 & 0.01 & 56076 & 0.01 & 3 \\
uniform 30x30 & 297819 & \textbf{0.02} & 298485 & 0.03 & 297874 & 0.05 & 4 \\
uniform 40x40 & 965875 & \textbf{0.06} & 967373 & 0.08 & 971010 & 0.13 & 5 \\
uniform 50x50 & 2415720 & \textbf{0.11} & 2414279 & 0.18 & 2419385 & 0.34 & 6 \\
uniform 60x60 & 5077348 & \textbf{0.23} & 5089275 & 0.33 & 5095460 & 0.77 & 9 \\
uniform 70x70 & 9578626 & \textbf{0.32} & 9561747 & 0.51 & 9549687 & 1.25 & 10 \\
uniform 80x80 & 16505833 & \textbf{0.59} & 16422705 & 0.93 & 16474525 & 1.87 & 10 \\
uniform 90x90 & 26650437 & \textbf{0.93} & 26726070 & 1.16 & 26706156 & 3.04 & 11 \\
uniform 100x100 & 41027445 & \textbf{1.12} & 41001387 & 1.89 & 41038180 & 4.78 & 14 \\
uniform 110x110 & 60512662 & \textbf{1.72} & 60549540 & 2.37 & 60508210 & 6.87 & 15 \\
uniform 120x120 & 86397256 & \textbf{2.08} & 86108044 & 3.23 & 86019130 & 10.47 & 18 \\
uniform 130x130 & 119380881 & \textbf{3.02} & 119421396 & 4.06 & 119417016 & 12.52 & 16 \\
uniform 140x140 & 161524589 & \textbf{3.58} & 161725915 & 5.6 & 161535754 & 16.97 & 18 \\
uniform 150x150 & 213377462 & \textbf{5.02} & 214064556 & 6.9 & 213453225 & 22.48 & 19 \\
normal 10x10 & 4037 & \textbf{0.0} & 3997 & 0.0 & 3997 & 0.0 & 2 \\
normal 20x20 & 66006 & \textbf{0.01} & 66372 & 0.01 & 66104 & 0.01 & 3 \\
normal 30x30 & 343319 & \textbf{0.02} & 342316 & 0.03 & 342776 & 0.05 & 3 \\
normal 40x40 & 1096961 & \textbf{0.06} & 1098741 & 0.09 & 1101256 & 0.17 & 7 \\
normal 50x50 & 2712329 & \textbf{0.12} & 2709929 & 0.2 & 2708557 & 0.38 & 8 \\
normal 60x60 & 5668986 & \textbf{0.21} & 5671907 & 0.33 & 5678451 & 0.72 & 8 \\
normal 70x70 & 10561145 & \textbf{0.42} & 10588835 & 0.57 & 10581535 & 1.29 & 10 \\
normal 80x80 & 18172093 & \textbf{0.51} & 18160338 & 0.87 & 18141092 & 2.22 & 12 \\
normal 90x90 & 29222387 & \textbf{0.91} & 29231041 & 1.3 & 29283340 & 2.84 & 10 \\
normal 100x100 & 44751122 & \textbf{1.31} & 44735031 & 1.72 & 44753417 & 5.22 & 15 \\
normal 110x110 & 65809366 & \textbf{1.64} & 65817524 & 2.39 & 65812802 & 6.97 & 15 \\
normal 120x120 & 93529513 & \textbf{2.26} & 93491028 & 3.58 & 93581308 & 8.65 & 14 \\
normal 130x130 & 129150096 & \textbf{3.26} & 129310194 & 4.14 & 129238943 & 12.84 & 17 \\
normal 140x140 & 174245361 & \textbf{3.75} & 174296950 & 5.91 & 174169032 & 20.14 & 21 \\
normal 150x150 & 230484514 & \textbf{4.28} & 230242366 & 7.32 & 230292305 & 24.21 & 21 \\
euclidean 10x10 & 5032 & \textbf{0.0} & 5015 & 0.0 & 5015 & 0.01 & 1 \\
euclidean 20x20 & 81714 & \textbf{0.01} & 81701 & 0.01 & 81701 & 0.01 & 2 \\
euclidean 30x30 & 424425 & \textbf{0.03} & 424261 & 0.04 & 424261 & 0.06 & 3 \\
euclidean 40x40 & 1331726 & \textbf{0.06} & 1330070 & 0.11 & 1330070 & 0.15 & 4 \\
euclidean 50x50 & 3342515 & \textbf{0.13} & 3337157 & 0.24 & 3337157 & 0.35 & 4 \\
euclidean 60x60 & 6637101 & \textbf{0.24} & 6622844 & 0.42 & 6622844 & 0.63 & 5 \\
euclidean 70x70 & 12373648 & \textbf{0.33} & 12345122 & 0.7 & 12345122 & 1.01 & 4 \\
euclidean 80x80 & 21088451 & \textbf{0.55} & 21060424 & 1.01 & 21060424 & 1.34 & 3 \\
euclidean 90x90 & 33842019 & \textbf{0.85} & 33831315 & 1.48 & 33831315 & 2.01 & 4 \\
euclidean 100x100 & 50386904 & \textbf{1.08} & 50351081 & 2.19 & 50350547 & 3.33 & 5 \\
\bottomrule
\end{tabular}}
\end{table}

%% file: tables/vnsconv2.tex
\clearpage
\newgeometry{margin=1cm}
\thispagestyle{empty}

\begin{sidewaystable*}[t!]
\centering
\small
\caption{Solution value, running time in seconds and number of iterations for \textit{2exOpt} and variations (convergence to local optima)}
\label{vnsconvt2}\scalebox{0.75}{\begin{tabular}{@{}lrlrlrlcrlc@{}}
\toprule
 & \multicolumn{2}{c}{2ExOptFirst} & \multicolumn{2}{c}{AA+2exOptFirst} & \multicolumn{3}{c}{AA2exOptStep} & \multicolumn{3}{c}{AA2exOptFirstStep}\\ \cmidrule(lr){2-3} \cmidrule(lr){4-5} \cmidrule(lr){6-8} \cmidrule(lr){9-11}
instances & value & time & value & time & value & time & iter & value & time & iter\\
\midrule
uniform 10x10 & 3156 & \textbf{0.02} & 3059 & 0.01 & 3054 & 0.02 & 2 & 3059 & 0.02 & 3 \\
uniform 20x20 & 54670 & \textbf{0.35} & 54877 & 0.24 & 54718 & 0.32 & 3 & 54431 & 0.2 & 4 \\
uniform 30x30 & 291902 & 2.27 & 294044 & \textbf{1.12} & 291184 & 2.64 & 4 & 290011 & 1.17 & 4 \\
uniform 40x40 & 948344 & 9.78 & 958550 & \textbf{3.29} & 953938 & 5.61 & 3 & 958215 & 3.4 & 3 \\
uniform 50x50 & 2379856 & 33.02 & 2399151 & 11.15 & 2392151 & 16.57 & 3 & 2395319 & \textbf{8.08} & 3 \\
uniform 60x60 & 5044883 & 64.73 & 5026000 & 36.08 & 5030618 & 35.22 & 3 & 5026865 & \textbf{20.45} & 4 \\
uniform 70x70 & 9479099 & 168.6 & 9511756 & 67.85 & 9521222 & 78.27 & 3 & 9501548 & \textbf{28.81} & 3 \\
uniform 80x80 & 16418360 & 252.23 & 16400987 & 120.41 & 16390406 & 132.04 & 3 & 16381373 & \textbf{56.49} & 3 \\
uniform 90x90 & 26507000 & 569.45 & 26499481 & 229.48 & 26536687 & 238.11 & 3 & 26546135 & \textbf{113.07} & 4 \\
uniform 100x100 & 40753550 & 878.32 & 40894844 & 293.74 & 40949795 & 184.26 & 1 & 40875653 & \textbf{155.41} & 3 \\
uniform 110x110 & 60079399 & 1539.8 & 60231687 & 458.67 & 60277301 & 487.4 & 3 & 60196421 & \textbf{317.3} & 4 \\
uniform 120x120 & 85818278 & 2120.85 & 85774789 & 1090.37 & 85996522 & 526.83 & 2 & 86070239 & \textbf{306.32} & 2 \\
uniform 130x130 & 118773110 & 3515.46 & 118967905 & 1105.4 & 119034719 & 827.06 & 2 & 119133276 & \textbf{452.11} & 3 \\
uniform 140x140 & 160780185 & 4860.32 & 160956538 & 1304.17 & 161002007 & 1479.93 & 3 & 161113803 & \textbf{764.84} & 3 \\
uniform 150x150 & 213525103 & 5514.74 & 213372569 & \textbf{538.34} & 213372569 & 748.16 & 1 & 213372569 & 553.21 & 1 \\
normal 10x10 & 3866 & \textbf{0.02} & 3895 & 0.01 & 3886 & 0.02 & 2 & 3917 & 0.01 & 2 \\
normal 20x20 & 65262 & \textbf{0.3} & 65137 & 0.28 & 65166 & 0.36 & 3 & 65258 & 0.21 & 4 \\
normal 30x30 & 338569 & 2.9 & 340096 & 1.19 & 340240 & 1.52 & 2 & 340534 & \textbf{0.86} & 3 \\
normal 40x40 & 1087006 & 10.28 & 1087569 & 6.04 & 1090323 & 6.93 & 3 & 1089412 & \textbf{3.46} & 4 \\
normal 50x50 & 2695007 & 26.39 & 2697747 & 14.44 & 2697124 & 19.13 & 3 & 2696860 & \textbf{7.45} & 3 \\
normal 60x60 & 5637608 & 71.64 & 5639469 & 34.18 & 5634802 & 45.69 & 4 & 5638741 & \textbf{18.75} & 3 \\
normal 70x70 & 10538891 & 159.53 & 10527751 & 61.85 & 10524931 & 80.22 & 3 & 10532494 & \textbf{33.81} & 3 \\
normal 80x80 & 18102861 & 292.68 & 18102161 & 145.45 & 18123379 & 148.5 & 4 & 18125319 & \textbf{62.56} & 3 \\
normal 90x90 & 29162243 & 447.82 & 29167487 & 166.29 & 29176575 & 193.61 & 3 & 29167084 & \textbf{102.4} & 3 \\
normal 100x100 & 44610176 & 953.0 & 44644532 & 272.9 & 44626268 & 376.46 & 4 & 44645246 & \textbf{153.74} & 3 \\
normal 110x110 & 65589378 & 1404.23 & 65635027 & 561.99 & 65669769 & 423.52 & 3 & 65646106 & \textbf{233.42} & 3 \\
normal 120x120 & 93315766 & 2071.35 & 93321138 & 697.7 & 93338052 & 692.75 & 3 & 93300933 & \textbf{346.7} & 3 \\
normal 130x130 & 128872342 & 3329.54 & 129005518 & 630.07 & 128978046 & 784.53 & 2 & 129030228 & \textbf{361.45} & 2 \\
normal 140x140 & 173877153 & 4669.47 & 174004558 & 1379.7 & 174104009 & 857.84 & 2 & 174117705 & \textbf{565.83} & 2 \\
normal 150x150 & 229879808 & 6572.92 & 229985798 & 2161.19 & 230286566 & 1481.59 & 3 & 230254077 & \textbf{757.09} & 3 \\
euclidean 10x10 & 4988 & \textbf{0.04} & 4995 & 0.02 & 4992 & 0.02 & 1 & 4996 & 0.02 & 1 \\
euclidean 20x20 & 81833 & 1.46 & 81644 & \textbf{0.33} & 81644 & 0.31 & 1 & 81644 & 0.28 & 1 \\
euclidean 30x30 & 424227 & 17.82 & 424425 & \textbf{1.63} & 424425 & 1.65 & 1 & 424425 & 1.64 & 1 \\
euclidean 40x40 & 1330114 & 84.25 & 1331592 & \textbf{7.63} & 1331592 & 7.28 & 1 & 1331592 & 7.24 & 1 \\
euclidean 50x50 & 3344106 & 347.38 & 3342208 & 22.61 & 3342208 & 20.49 & 1 & 3342208 & \textbf{18.78} & 1 \\
euclidean 60x60 & 6628784 & 968.15 & 6637101 & \textbf{43.81} & 6637101 & 43.91 & 1 & 6637101 & 43.81 & 1 \\
euclidean 70x70 & 12343342 & 2404.75 & 12373648 & \textbf{90.12} & 12373648 & 90.34 & 1 & 12373648 & 90.1 & 1 \\
euclidean 80x80 & 21098260 & 5579.32 & 21088451 & \textbf{174.46} & 21088451 & 174.94 & 1 & 21088451 & 174.98 & 1 \\
euclidean 90x90 & 33892498 & 11440.65 & 33841998 & 333.66 & 33841998 & 338.05 & 1 & 33841998 & \textbf{326.42} & 1 \\
euclidean 100x100 & 50313528 & 19808.73 & 50386904 & \textbf{514.2} & 50386904 & 515.13 & 1 & 50386904 & 514.74 & 1 \\
\bottomrule
\end{tabular}}
\end{sidewaystable*}
\clearpage
\restoregeometry

%% file: tables/vnsms.tex
\clearpage
\newgeometry{margin=1cm}
\thispagestyle{empty}

\begin{sidewaystable*}[t!]
\centering
\small
\caption{Solution value, running time in seconds and number of iterations for Variable Neighborhood Search and multi-start \textit{AA}}
\label{vnsmst}\scalebox{0.8}{\begin{tabular}{@{}lrlcrlcrlcrlcrlccc@{}}
\toprule
 & \multicolumn{3}{c}{AA2ExOptFirstStep} & \multicolumn{3}{c}{4AA2ExOptFirstStep} & \multicolumn{3}{c}{10AA2ExOptFirstStep} & \multicolumn{3}{c}{100AA2ExOptFirstStep} & \multicolumn{5}{c}{msAA}\\ \cmidrule(lr){2-4} \cmidrule(lr){5-7} \cmidrule(lr){8-10} \cmidrule(lr){11-13} \cmidrule(lr){14-18}
instances & value & time & iter & value & time & iter & value & time & iter & value & time & iter & value & time & iter & best iter & $\sigma(\text{best iter})$\\
\midrule
uniform 10x10 & 3162 & 0.02 & 3 & 3126 & 0.01 & 1 & 3025 & 0.01 & 1 & \textbf{2983} & 0.06 & 1 & 2995 & 0.07 & 116 & 47 & 41\\
uniform 20x20 & 55131 & 0.15 & 2 & 54601 & 0.17 & 2 & 54294 & 0.19 & 2 & \textbf{53281} & 0.58 & 1 & 53620 & 0.59 & 131 & 54 & 37\\
uniform 30x30 & 293385 & 0.89 & 3 & 292039 & 0.83 & 2 & 289483 & 0.92 & 2 & \textbf{286542} & 2.42 & 1 & 287169 & 2.44 & 130 & 57 & 49\\
uniform 40x40 & 955295 & 3.03 & 3 & 950608 & 2.77 & 3 & 951947 & 2.87 & 2 & 942849 & 6.82 & 1 & \textbf{939052} & 6.85 & 138 & 89 & 32\\
uniform 50x50 & 2380817 & 11.35 & 5 & 2379835 & 11.88 & 4 & 2375551 & 7.42 & 2 & 2370805 & 15.35 & 1 & \textbf{2360529} & 16.56 & 165 & 83 & 52\\
uniform 60x60 & 5038934 & 19.96 & 3 & 5030082 & 15.28 & 2 & 5015756 & 18.16 & 2 & \textbf{4990868} & 35.36 & 2 & 4993774 & 38.42 & 208 & 112 & 35\\
uniform 70x70 & 9479825 & 34.21 & 4 & 9436974 & 43.32 & 3 & 9445502 & 39.85 & 3 & 9413893 & 54.29 & 1 & \textbf{9399736} & 61.76 & 203 & 115 & 67\\
uniform 80x80 & 16389632 & 61.47 & 3 & 16357168 & 55.61 & 2 & 16303348 & 59.12 & 2 & \textbf{16261295} & 95.21 & 1 & 16264848 & 104.0 & 217 & 95 & 54\\
uniform 90x90 & 26505894 & 110.55 & 3 & 26456700 & 94.5 & 3 & 26407075 & 80.08 & 1 & 26356116 & 151.23 & 2 & \textbf{26342919} & 160.45 & 226 & 83 & 64\\
uniform 100x100 & 40782492 & 141.59 & 3 & 40712949 & 180.44 & 3 & 40633567 & 165.63 & 3 & 40540438 & 208.3 & 1 & \textbf{40506423} & 241.3 & 241 & 116 & 96\\
uniform 120x120 & 85825930 & 342.18 & 3 & 85579139 & 274.87 & 2 & 85471530 & 333.39 & 3 & 85335239 & 441.49 & 1 & \textbf{85283242} & 509.31 & 273 & 122 & 71\\
uniform 140x140 & 160605657 & 693.67 & 3 & 160415349 & 555.54 & 3 & 160292924 & 474.41 & 1 & 160035009 & 719.05 & 1 & \textbf{159912990} & 927.88 & 286 & 131 & 124\\
uniform 160x160 & 277129402 & 909.79 & 2 & 276565751 & 918.66 & 2 & 276159588 & 908.23 & 2 & \textbf{275721038} & 1386.9 & 1 & 275725334 & 1657.71 & 302 & 154 & 100\\
normal 10x10 & 3894 & 0.02 & 3 & 3855 & 0.01 & 2 & 3855 & 0.01 & 1 & \textbf{3808} & 0.07 & 1 & 3809 & 0.07 & 117 & 40 & 37\\
normal 20x20 & 65712 & 0.15 & 2 & 65077 & 0.17 & 2 & 64803 & 0.2 & 1 & \textbf{64293} & 0.58 & 1 & 64477 & 0.58 & 130 & 73 & 48\\
normal 30x30 & 338547 & 1.17 & 5 & 337693 & 0.95 & 3 & 338138 & 0.79 & 1 & \textbf{335113} & 2.75 & 2 & 335756 & 2.76 & 145 & 74 & 42\\
normal 40x40 & 1090670 & 2.81 & 3 & 1088357 & 3.1 & 3 & 1085519 & 2.69 & 2 & \textbf{1081375} & 7.56 & 1 & 1082915 & 7.58 & 154 & 81 & 46\\
normal 50x50 & 2696368 & 8.24 & 3 & 2692035 & 8.33 & 2 & 2682121 & 8.66 & 3 & \textbf{2678345} & 17.52 & 2 & 2680271 & 17.58 & 175 & 71 & 55\\
normal 60x60 & 5647247 & 17.06 & 3 & 5633194 & 14.77 & 1 & 5627675 & 17.07 & 2 & \textbf{5616899} & 31.56 & 2 & 5617125 & 32.18 & 173 & 83 & 55\\
normal 70x70 & 10549768 & 26.89 & 1 & 10519922 & 34.7 & 3 & 10509205 & 30.19 & 2 & \textbf{10493809} & 57.37 & 2 & 10494503 & 61.86 & 201 & 104 & 64\\
normal 80x80 & 18095404 & 72.05 & 3 & 18069406 & 59.64 & 2 & 18067347 & 55.46 & 2 & 18032081 & 86.61 & 1 & \textbf{18023497} & 100.11 & 209 & 112 & 62\\
normal 90x90 & 29115217 & 107.77 & 3 & 29103538 & 103.37 & 2 & 29097191 & 95.29 & 2 & 29045978 & 165.73 & 2 & \textbf{29027250} & 187.3 & 264 & 120 & 71\\
normal 100x100 & 44618697 & 130.7 & 2 & 44578918 & 138.0 & 2 & 44556729 & 162.61 & 3 & 44484747 & 245.72 & 3 & \textbf{44482231} & 279.76 & 274 & 172 & 62\\
normal 120x120 & 93293438 & 343.2 & 3 & 93162243 & 313.92 & 2 & 93112300 & 309.4 & 2 & 93023046 & 506.08 & 2 & \textbf{92984865} & 540.0 & 282 & 149 & 93\\
normal 140x140 & 173820624 & 535.5 & 2 & 173653510 & 510.49 & 2 & 173594266 & 481.53 & 1 & 173434718 & 815.2 & 2 & \textbf{173430869} & 900.03 & 279 & 144 & 76\\
normal 160x160 & 298434202 & 967.33 & 2 & 297840806 & 899.65 & 2 & 297816150 & 1030.84 & 2 & 297540220 & 1211.89 & 1 & \textbf{297480023} & 1567.93 & 294 & 126 & 62\\
euclidean 10x10 & 5037 & 0.02 & 1 & \textbf{5026} & 0.02 & 1 & 5027 & 0.02 & 1 & 5026 & 0.11 & 1 & 5026 & 0.11 & 116 & 6 & 7\\
euclidean 20x20 & 82675 & 0.25 & 1 & 82008 & 0.26 & 1 & 81842 & 0.31 & 1 & \textbf{81718} & 1.0 & 1 & 81718 & 1.0 & 129 & 12 & 11\\
euclidean 30x30 & 411014 & 1.78 & 1 & 408739 & 1.72 & 1 & 407379 & 1.91 & 1 & \textbf{406970} & 4.23 & 1 & 406970 & 4.24 & 162 & 32 & 43\\
euclidean 40x40 & 1348302 & 6.68 & 1 & 1342159 & 6.99 & 1 & 1339683 & 7.09 & 1 & 1337792 & 12.69 & 1 & \textbf{1337738} & 12.72 & 204 & 48 & 58\\
euclidean 50x50 & 3231060 & 21.05 & 1 & 3219207 & 20.39 & 1 & 3214867 & 19.94 & 1 & 3210442 & 30.74 & 1 & \textbf{3210280} & 31.97 & 254 & 37 & 36\\
euclidean 60x60 & 6548901 & 44.42 & 1 & 6519075 & 44.82 & 1 & 6515800 & 46.24 & 1 & 6507833 & 65.26 & 1 & \textbf{6507813} & 65.41 & 304 & 32 & 23\\
euclidean 70x70 & 12315235 & 93.93 & 1 & 12283239 & 100.51 & 1 & 12264197 & 96.28 & 1 & 12257619 & 126.03 & 1 & \textbf{12256435} & 128.94 & 388 & 74 & 76\\
euclidean 80x80 & 21240164 & 187.89 & 1 & 21143316 & 183.3 & 1 & 21104571 & 185.35 & 1 & 21096255 & 229.53 & 1 & \textbf{21095365} & 232.0 & 459 & 144 & 132\\
euclidean 90x90 & 33385322 & 335.48 & 1 & 33323860 & 319.99 & 1 & 33296502 & 326.28 & 1 & 33279588 & 388.9 & 1 & \textbf{33277417} & 398.29 & 558 & 81 & 126\\
euclidean 100x100 & 51524424 & 530.7 & 1 & 51382552 & 535.98 & 1 & 51303227 & 538.1 & 1 & 51289100 & 632.49 & 1 & \textbf{51286565} & 633.16 & 597 & 158 & 133\\
euclidean 120x120 & 105192868 & 1291.27 & 1 & 105092433 & 1284.2 & 1 & 105037756 & 1404.01 & 1 & 104969850 & 1456.4 & 1 & \textbf{104965462} & 1556.45 & 908 & 93 & 112\\
\bottomrule
\end{tabular}}
\end{sidewaystable*}
\clearpage
\restoregeometry

%% file: main.bbl
\begin{thebibliography}{10}

\bibitem{adams2007level}
Warren~P Adams, Monique Guignard, Peter~M Hahn, and William~L Hightower.
\newblock A level-2 reformulation--linearization technique bound for the
  quadratic assignment problem.
\newblock {\em European Journal of Operational Research}, 180(3):983--996,
  2007.

\bibitem{ahuja2002survey}
Ravindra~K Ahuja, {\"O}zlem Ergun, James~B Orlin, and Abraham~P Punnen.
\newblock A survey of very large-scale neighborhood search techniques.
\newblock {\em Discrete Applied Mathematics}, 123(1):75--102, 2002.

\bibitem{ahuja2007very}
Ravindra~K Ahuja, {\"O}zlem Ergun, James~B Orlin, and Abraham~P Punnen.
\newblock Very large-scale neighborhood search: Theory, algorithms, and
  applications.
\newblock {\em Handbook of Approximation Algorithms and Metaheuristics}, 10,
  2007.

\bibitem{angel1998quality}
Eric Angel and Vassilis Zissimopoulos.
\newblock On the quality of local search for the quadratic assignment problem.
\newblock {\em Discrete Applied Mathematics}, 82(1-3):15--25, 1998.

\bibitem{burkard2012assignment}
Rainer Burkard, Mauro Dell'Amico, and Silvano Martello.
\newblock {\em Assignment problems: revised reprint}.
\newblock SIAM, 2012.

\bibitem{burkard1998quadratic}
Rainer~E Burkard, Eranda Cela, G{\"u}nter Rote, and Gerhard~J Woeginger.
\newblock The quadratic assignment problem with a monotone anti-monge and a
  symmetric toeplitz matrix: Easy and hard cases.
\newblock {\em Mathematical Programming}, 82(1):125--158, 1998.

\bibitem{cela2013quadratic}
Eranda Cela.
\newblock {\em The quadratic assignment problem: theory and algorithms},
  volume~1.
\newblock Springer Science \& Business Media, 2013.

\bibitem{custic2017average}
Ante {\'C}usti{\'c} and Abraham~P Punnen.
\newblock Average value of solutions of the bipartite quadratic assignment
  problem and linkages to domination analysis.
\newblock {\em Operations Research Letters}, 45(3):232--237, 2017.

\bibitem{custicbilinear}
Ante {\'C}usti{\'c}, Vladyslav Sokol, Abraham~P Punnen, and Binay Bhattacharya.
\newblock The bilinear assignment problem: complexity and polynomially solvable
  special cases.
\newblock {\em Mathematical Programming}, pages 1--21, 2017.

\bibitem{feo1989probabilistic}
Thomas~A Feo and Mauricio~GC Resende.
\newblock A probabilistic heuristic for a computationally difficult set
  covering problem.
\newblock {\em Operations Research Letters}, 8(2):67--71, 1989.

\bibitem{fon1997arrays}
Dmitri~G Fon-Der-Flaass.
\newblock Arrays of distinct representatives—a very simple np-complete
  problem.
\newblock {\em Discrete Mathematics}, 171(1-3):295--298, 1997.

\bibitem{frieze1983complexity}
Alan~M Frieze.
\newblock Complexity of a 3-dimensional assignment problem.
\newblock {\em European Journal of Operational Research}, 13(2):161--164, 1983.

\bibitem{frieze1989algorithms}
Alan~M Frieze, J~Yadegar, S~El-Horbaty, and D~Parkinson.
\newblock Algorithms for assignment problems on an array processor.
\newblock {\em Parallel Computing}, 11(2):151--162, 1989.

\bibitem{glover1997travelling}
Fred Glover and Abraham~P Punnen.
\newblock The travelling salesman problem: new solvable cases and linkages with
  the development of approximation algorithms.
\newblock {\em Journal of the Operational Research Society}, 48(5):502--510,
  1997.

\bibitem{glover2015integrating}
Fred Glover, Tao Ye, Abraham~P Punnen, and Gary Kochenberger.
\newblock Integrating tabu search and vlsn search to develop enhanced
  algorithms: A case study using bipartite boolean quadratic programs.
\newblock {\em European Journal of Operational Research}, 241(3):697--707,
  2015.

\bibitem{hansen2016variable}
Pierre Hansen, Nenad Mladenovi{\'c}, Raca Todosijevi{\'c}, and Sa{\"\i}d
  Hanafi.
\newblock Variable neighborhood search: basics and variants.
\newblock {\em EURO Journal on Computational Optimization}, pages 1--32, 2016.

\bibitem{hart1987semi}
J~Pirie Hart and Andrew~W Shogan.
\newblock Semi-greedy heuristics: An empirical study.
\newblock {\em Operations Research Letters}, 6(3):107--114, 1987.

\bibitem{karapetyan2012heuristic}
Daniel Karapetyan and Abraham~P Punnen.
\newblock Heuristic algorithms for the bipartite unconstrained 0-1 quadratic
  programming problem.
\newblock {\em arXiv preprint arXiv:1210.3684}, 2012.

\bibitem{kaufman1978algorithm}
L~Kaufman and Fernand Broeckx.
\newblock An algorithm for the quadratic assignment problem using bender's
  decomposition.
\newblock {\em European Journal of Operational Research}, 2(3):207--211, 1978.

\bibitem{konno1980maximizing}
Hiroshi Konno.
\newblock Maximizing a convex quadratic function over a hypercube.
\newblock {\em Journal of the Operations Research Society of Japan},
  23(2):171--189, 1980.

\bibitem{kuhn1955hungarian}
Harold~W Kuhn.
\newblock The hungarian method for the assignment problem.
\newblock {\em Naval Research Logistics Quarterly}, 2(1-2):83--97, 1955.

\bibitem{lawler1963quadratic}
Eugene~L Lawler.
\newblock The quadratic assignment problem.
\newblock {\em Management Science}, 9(4):586--599, 1963.

\bibitem{pardalos1994quadratic}
Panos~M Pardalos, Henry Wolkowicz, et~al.
\newblock {\em Quadratic Assignment and Related Problems: DIMACS Workshop, May
  20-21, 1993}, volume~16.
\newblock American Mathematical Soc., 1994.

\bibitem{pierskalla1968letter}
William~P Pierskalla.
\newblock Letter to the editor-the multidimensional assignment problem.
\newblock {\em Operations Research}, 16(2):422--431, 1968.

\bibitem{punnen2015average}
Abraham~P Punnen, Piyashat Sripratak, and Daniel Karapetyan.
\newblock Average value of solutions for the bipartite boolean quadratic
  programs and rounding algorithms.
\newblock {\em Theoretical Computer Science}, 565:77--89, 2015.

\bibitem{punnen2016bipartite}
Abraham~P Punnen and Yang Wang.
\newblock The bipartite quadratic assignment problem and extensions.
\newblock {\em European Journal of Operational Research}, 250(3):715--725,
  2016.

\bibitem{sahni1976p}
Sartaj Sahni and Teofilo Gonzalez.
\newblock P-complete approximation problems.
\newblock {\em Journal of the ACM (JACM)}, 23(3):555--565, 1976.

\bibitem{torki1996low}
Abdolhamid Torki, Yatsutoshi Yajima, and Takao Enkawa.
\newblock A low-rank bilinear programming approach for sub-optimal solution of
  the quadratic assignment problem.
\newblock {\em European Journal of Operational Research}, 94(2):384--391, 1996.

\bibitem{tsui1990microcomputer}
Louis~Y Tsui and Chia-Hao Chang.
\newblock A microcomputer based decision support tool for assigning dock doors
  in freight yards.
\newblock {\em Computers \& Industrial Engineering}, 19(1-4):309--312, 1990.

\bibitem{tsui1992optimal}
Louis~Y Tsui and Chia-Hao Chang.
\newblock An optimal solution to a dock door assignment problem.
\newblock {\em Computers \& Industrial Engineering}, 23(1-4):283--286, 1992.

\end{thebibliography}
